\documentclass[12pt,a4paper,UTF8]{article}
\usepackage{amsmath,mathtools,amsfonts,amsmath,amssymb,mathrsfs,dsfont,cases,amsthm,bm,bbm}
\usepackage[top=2.5cm, bottom=2.5cm, left=2.5cm, right=2.5cm]{geometry}
\usepackage{float}
\usepackage{authblk}
\usepackage{enumerate,enumitem}
\usepackage[colorlinks=true,allcolors=blue]{hyperref}%
\usepackage{graphicx,subfigure}
\graphicspath{{figure/}}
\usepackage[authoryear]{natbib}
\usepackage{array,booktabs,makecell,tabularx,multirow,diagbox,rotating,threeparttable}  
\usepackage{algorithm,algorithmic}
\newtheoremstyle{thry}
{3pt}
{3pt}
{\itshape}
{}
{\scshape\bfseries}
{:}
{.5em}
{}
\theoremstyle{thry}

\newcommand{\bR}{\mathbb{R}}
\newcommand{\bP}{\mathbb{P}}
\newcommand{\bE}{\mathbb{E}}

\newcommand{\bbeta}{\bm{\beta}}
\newcommand{\bx}{\bm{x}}
\newcommand{\tr}{\mathrm{tr}}
\newcommand{\var}{\mathrm{var}}
\newcommand{\diag}{\mathrm{diag}}
\newcommand{\tabincell}[2]{\begin{tabular}
	{@{}#1@{}}#2\end{tabular}}

\newtheorem{theorem}{\indent Theorem}
\newtheorem{lemma}{\indent  Lemma}
\newtheorem{assumption}{\indent  Assumption}
\newtheorem{proposition}{\indent  Proposition}

\numberwithin{equation}{section}

\makeatletter
\renewcommand*{\@fnsymbol}[1]{\ensuremath{\ifcase#1
\or \dagger \or \ddagger \or \mathsection 
\or \mathparagraph \or \| \or \dagger\dagger
\or \ddagger\ddagger \else \@ctrerr\fi}}
\makeatother

\begin{document}

\title{Hypothesis testing for general network models}
\author{
\textsc{Kang Fu}\thanks{School of Mathematics and Statistics, and Key Laboratory of Nonlinear Analysis \& Applications (Ministry of Education), Central China Normal University, Wuhan 430079, China. \texttt{Email:} \textit{fukang@mails.ccnu.edu.cn}.},
\hspace{1mm}
\textsc{Jianwei Hu}\thanks{School of Mathematics and Statistics, and Hubei Key Laboratory of Mathematical Sciences, Central China Normal University, Wuhan 430079, China. \texttt{Email:} \textit{jwhu@ccnu.edu.cn}.},
\hspace{1mm}
and
\hspace{1mm}
\textsc{Seydou Keita}\thanks{School of Mathematics and Statistics, Central China Normal University, Wuhan 430079, China. \texttt{Email:} \textit{badco62003@yahoo.fr}.}\\
\textit{Central China Normal University}
}
\date{\empty} 

\maketitle

\vspace*{-1cm}
\begin{abstract}
The network data has attracted considerable attention in modern statistics. In research on complex network data, one key issue is finding its underlying connection structure given a network sample. The methods that have been proposed in literature usually assume that the underlying structure is a known model. In practice, however, the true model is usually unknown, and network learning procedures based on these methods may suffer from model misspecification. To handle this issue, based on the random matrix theory, we first give a spectral property of the normalized adjacency matrix under a mild condition. Further, we establish a general goodness-of-fit test procedure for the unweight and undirected network. We prove that the null distribution of the proposed statistic converges in distribution to the standard normal distribution. Theoretically, this testing procedure is suitable for nearly all popular network models, such as stochastic block models, and latent space models. Further, we apply the proposed method to the degree-corrected mixed membership model and give a sequential estimator of the number of communities. Both simulation studies and real-world data examples indicate that the proposed method works well.

\textbf{Key words:} Hypothesis testing; Network data; Normal distribution; Spectral method; Stochastic block model; Wigner matrix
\end{abstract}

\section{Introduction}

Network data appear in many disciplines, such as sociology, biology, computer science, and many others \citep{Scott:2000,Guimera:2005}. A network usually represents a relationship among a collection of individuals, such as protein networks and social relationship networks. In general, a network $\mathcal{G}$ with $n$ nodes can be represented by a corresponding adjacency matrix $A\in\bR^{n\times n}$, where $(i,j)$-entry of $A$ represents the link relationship between node $i$ and node $j$. For the unweighted network, $A_{ij}=1$ if there is a link from node $i$ to node $j$ and $A_{ij}=0$ otherwise. In our study, we mainly focus on the undirected and unweighted network, that is, $A$ is a symmetric and binary matrix.

There are various studies on complex network data, and a majority of network models have been proposed, such as the Erd\H{o}s-R\'{e}nyi (E-R) model \citep{Erdos:2012,Erdos:2013}, the $\beta$-model \citep{Chatterjee:2011}, the stochastic block model (SBM) \citep{Holland:1983}, the degree-corrected stochastic block model (DCSBM) \citep{Karrer:2011}, the degree-corrected mixed membership (DCMM) model \citep{Jin:2023}, and so on. In the past decades, network data analysis mainly depends on these classical models. Under a given model, implementing statistic inference for network data is a popular interest of research. In addition, hypothesis testing is another research hot-spot in network data analysis, especially in the SBM and its variants. For the network with community structure, hypothesis testing was initially used to test whether network data has a community structure 	\citep{Bickel:2016,Cammarata:2023}. Later, under the framework of the SBM, these methods have also been extended to estimate the number of communities. Specifically, given an adjacency matrix $A$, the basic idea is to consider the hypothesis test problem $K=K_0$, where $K$ and $K_0$ are the true and hypothesis number of communities, respectively. Based on the largest singular value of a residual matrix, \cite{Lei:2016} proposed a goodness-of-fit test for the SBM and extended this theory to estimate the number of communities by the sequential method. Similarly, \cite{Hu:2021} also investigated the goodness-of-fit test for the SBM. They considered the maximum entry-wise deviation between the adjacency matrix and the corresponding edge probability matrix. Further, \cite{Wu:2022} proposed a new statistic to investigate the goodness-of-fit for SBMs by introducing the local smoothing technology, and the statistic is adapted to the case of a small-sized community with an unbalanced community. Under the framework of DCMM models, \cite{Jin:2021} and \cite{Cammarata:2023} considered the global testing problem, i.e., whether an undirected network has one or multiple communities. Further, \cite{Fan:2022} considered the testing of whether two nodes share a common community membership. As a general case, \cite{Du:2023} considered the hypothesis testing problem that two vertices $i$ and $j$  have the same latent positions under generalized random dot product graphs. As introduced later, the DCMM model is a specified case of the generalized random dot product graph. Hence, the results in \cite{Du:2023} are generalizations of the corresponding results in \cite{Fan:2022}.

All the literature mentioned above considers one-sample scenarios. In the hypothesis test of the network analysis, another issue is the two-sample test, that is, whether two network samples are generated from the same network model. Under the random dot product graph, based on the kernel function, \cite{Tang:2017} proposed a testing method to justify whether two independent finite-dimensional samples have a common population model. Further \cite{Ghoshdastidar:2020} proposed two test statistics using the Frobenius norm and spectral norm. \cite{Chen:2021} used the trace of a normalized matrix to obtain the statistic and proved that the null distribution is the standard normal distribution. Under the framework of SBMs, \cite{Fu:2024} and \cite{Fu:2023} extended one-sample testing methods to the case of two samples, and proposed two statistics to test whether two samples have the same community structure.

Notice the studies mentioned above are based on a known model. In statistical learning, the true model is usually unknown. Hence, it is significant to choose an appropriate model for network learning. In this article, we consider constructing a general framework of the goodness-of-fit for network models. For a general network, we give a spectral property of the normalized adjacency matrix under a mild condition, i.e., the trace of the third order of the normalized adjacency matrix converges in distribution to a normal distribution. It is worth noting that our result is a nontrivial conclusion, including the results of \cite{Dong:2020} and \cite{Wu:2024} as special cases since they only consider SBMs. The main contribution is twofold. First, by the eigen-decomposition, we use a new technology strategy to prove the spectral property, which only needs a weaker condition for the estimators $\hat{p}_{ij}$'s. Second, based on this spectral property, we propose a goodness-of-fit test procedure for nearly all existing network models, such as $\beta$-models, stochastic block models, and latent space models. Compared with the test procedure in \cite{Lei:2016} and \cite{Hu:2021}, the proposed statistic converges to a normal distribution fast and does not require a bootstrap correction process. Meanwhile, the proposed test procedure is suitable for more general models, not limited to stochastic block models. Further, we apply the proposed method to DCMM models and propose an empirically estimated method by sequentially using the proposed goodness-of-fit test procedure. Empirically, we also find that the sequential testing estimation works well.

The remainder of the article is organized as follows. In Section \ref{sec:method}, we introduce the basic backgrounds of some common network models. The spectral property of the adjacency and a goodness-of-fit test procedure are also given in this Section. In Section \ref{sec:appdcmm}, we apply the proposed method to estimate the number of communities in DCMM models. Simulation studies and real-world data examples are given in Sections \ref{sec:simulation} and \ref{sec:real}, respectively. All technical proofs are postponed to the Appendix.

\section{Model and methods}\label{sec:method}

In this section, we first introduce some classical network models and then give a general goodness-of-fit framework.
\subsection{Network models}\label{sec:model}
Before formally introducing models, we introduce some notations. For a matrix $A\in\bR^{n\times n}$, we use $\tr(A)$ and $\diag(A)$ to denote the trace of matrix and diagonal matrix with diagonal elements $(A_{11},\ldots,A_{nn})$. For a vector $\theta=(\theta_1,\ldots,\theta_n)$, let $\diag(\theta)$ be a diagonal matrix with diagonal elements $(\theta_1,\ldots,\theta_n)$. We use $\bm{1}_n$ and $I_n$ to denote the $n$-dimensional vector with all entries 1 and $n$-dimensional identical matrix. The notation $I[\cdot]$ is the indicator function. 
For a sequence of random variables $X_n$ and a positive sequence $a_n$, we write $X_n=O_p(a_n)$ if for any $\varepsilon>0$, there exists finite $M>0$ and $N>0$ such that $ \forall n>N, \bP\{|X_n/a_n|\geq M\}<\varepsilon$. We also write $X_n=o_p(a_n)$ if for any $\varepsilon>0$, $\bP\{\vert X_n/a_n\vert\geq\varepsilon\}\rightarrow 0$.

\textit{\textbf{Erd\H{o}s-R\'{e}nyi model.}} The Erd\H{o}s-R\'{e}nyi model proposed by \cite{Erdos:1957} is the most basic model in network data analysis. The model assumes that there is an edge between any pairs of nodes $(i,j)$ with probability $p$. Suppose that $A\in\{0,1\}^{n\times n}$ is an adjacency matrix of undirected network $\mathcal{G}$. Throughout this paper, we assume that the self-loops are not allowed, i.e., $A_{ii}=0$ for $1\leq i\leq n$. Hence, for an adjacency matrix $A$ from the E-R graph, the $(i,j)$-entry of $A$ follows the Bernoulli distribution with probability $p$. In practice, the link probability $p$ is usually unknown. A simple method to estimate $p$ is calculating the proportion of pairs of nodes that form an edge, that is, $\hat{p}=\sum_{i\neq j}A_{ij}/(n(n-1))$.

\textit{\textbf{$\beta$-model.}} The $\beta$-model, proposed by \cite{Chatterjee:2011}, is a special case of a class of models known as node-parameter models, where each node degree is associated with a corresponding parameter. For an undirected network with $n$ nodes, the $\beta$-model assume that the edge between nodes $i$ and $j$ exists with probability 
\[
p_{ij} = \dfrac{e^{\beta_i+\beta_j}}{1+e^{\beta_i+\beta_j}},
\]
independently of all other edges, where $\beta_i$ is the node parameter (also known asthe ``attractiveness" of vertex) of node $i$. It is not difficult to see that the probability connecting the node $i$ and node $j$ only depends on the parameter of the node $i$ and node $j$. When all $\beta_i$'s are equal to each other, the $\beta$-model naturally degenerates to the E-R model. Since the $\beta$-model can simply capture important features of real-world networks, the $\beta$-model, and its variations have been studied widely \citep{Yan:2013,Rinaldo:2013,Mukherjee:2018}. Under the framework of the $\beta$-model, let $d_i=\sum_{j\neq i}A_{ij}$ be the degree of the node $i$. Then, the likelihood function can be written as:
\[
l(\bbeta|A) = \dfrac{e^{\sum_i\beta_id_i}}{\prod_{i<j}(1+e^{\beta_i+\beta_j})}.
\]

Denote $\hat{\bbeta}=\mathop{\arg \min}\limits_{\bbeta}\log l(\bbeta|A)$ as the maximum likelihood estimator (MLE). The MLE can be obtained by solving the following equations:
\begin{equation}\label{eq:MLE}
	d_i=\sum_{j\neq i}\dfrac{e^{\hat{\beta}_i+\hat{\beta}_j}}{1+e^{\hat{\beta}_i+\hat{\beta}_j}}, (i=1,\ldots,n).
\end{equation}
\cite{Chatterjee:2011} established the consistency of $\hat{\bbeta}$. Specifically, let $L_n=\max_{i}|\beta_i|$, then there is a constant $C(L_n)$ depending only on $L_n$ such that $\bP\{\max_{1\leq i\leq n}|\hat{\beta}_i-\beta_i|\leq C(L_n)\sqrt{n^{-1}\log n}\}\geq 1-C(L_n)n^{-2}$. Further, by approximating the inverse of the Fisher information matrix, \cite{Yan:2013} proved the asymptotic normality of $\hat{\bbeta}$. Then, \cite{Rinaldo:2013} gave the necessary and sufficient conditions for the existence and uniqueness of $\hat{\bbeta}$.

\textit{\textbf{Stochastic block model.}} The stochastic block model was first proposed by \cite{Holland:1983}, and is usually used to model the network with community structure. Compared to the E-R model, a typical characteristic of SBMs is that nodes have a distinct community structure and the link probability between nodes only depends on the communities that they belong to. Formally, under the setting of the SBM, the $n$ nodes are clustered to $K$ disjoint sets, $\mathcal{C}_1,\ldots,\mathcal{C}_K$. Then, the link probability between nodes $i$ and $j$ is $p_{ij} = B_{\sigma_i\sigma_j}$, where $B\in[0,1]^{K\times K}$ is a $K\times K$ probability matrix and $\sigma_i=k$ if $i\in\mathcal{C}_k$. Write $Z\in\bR^{n\times K}$ be the membership matrix such that $Z_{ik}=1$ if $\sigma_i=k$ and $Z_{ik}=0$ otherwise. Then, we have
\[
\bE\{A\} = P - \diag(P),\ \text{with}\ P = ZBZ^\top.
\]
In the SBM, the main research issues are model selection and community detection. The goal of model selection is to estimate the number of communities $K$. The main methods to estimate $K$ include the sequential test \citep{Lei:2016,Hu:2021} and the likelihood-based method \citep{Saldna:2017,Wang:2017,Hu:2020}. The community detection aims to cluster all nodes into different communities such that the nodes in the same community have the same link behavior. The majority of methods have also been proposed to recover the community structure, such as spectral clustering \citep{Rohe:2011,Jin:2015}, pseudo-likelihood maximization \citep{Amini:2013}, and profile-pseudo likelihood methods \citep{Wang:2023,Fu:2023-2}. However, a limitation of the SBM is that the model assumes that all nodes are stochastically equivalent. In the real network, there are some nodes with `hubs' or high-degree and some nodes with low-degree, that is, heterogeneous. To address this shortcoming, \cite{Karrer:2011} proposed the degree-corrected stochastic block model. Similar to the SBM, the DCSBM replaces the link probability $B_{\sigma_i\sigma_j}$ with $\theta_i\theta_jB_{\sigma_i\sigma_j}$, where $\theta_i$ is the degree parameter associated with node $i$. Denote $\Theta=\diag(\theta_1,\ldots,\theta_n)$, we have
\[
\bE\{A\} = P - \diag(P),\ \text{with}\ P = \Theta ZBZ^\top\Theta.
\]
For the DCSBM, the corresponding methods of statistical inference have been proposed as the extension of the SBM.

\textit{\textbf{Degree-corrected mixed membership model.}} The degree-corrected mixed membership model, proposed by \cite{Jin:2023} is also a typical network model with a community structure. Unlike the SBM and DCSBM, the DCMM model allows for the node to belong to multiple communities. Specifically, in the DCMM model, the network also has $K$ communities. Each node has a membership vector $\pi_i=(\pi_i(1),\ldots,\pi_i(K))^\top$, where $\pi_i(k)$ is the weight that node $i$ belongs to community $k$, satisfying $\sum_{k}\pi_i(k)=1$ for all $i$. Similar to the DCSBM, each node also has a degree parameter $\theta_i$ in the DCMM model. Let $B\in[0,1]^{K\times K}$ be a symmetric probability matrix. Recall that $A$ is the adjacency matrix of a network, the DCMM model assumes that $A_{ij}$ is a Bernoulli random variable with probability
\[
p_{ij} = \theta_i\pi_i^\top B\pi_j\theta_j = \theta_i\theta_j\sum_{kl}\pi_i(k)B_{kl}\pi_j(l),\quad \text{for any}\ 1\leq i< j\leq n.
\]
Write $\Theta=\diag(\theta_1,\ldots,\theta_n)$ and $\Pi=(\pi_1,\ldots,\pi_n)^\top$ be a $n\times K$ membership matrix. Then, we have
\[
\bE\{A\} = P - \diag(P),\ \text{with}\ P = \Theta\Pi B\Pi^\top\Theta.
\]

It is not difficult to see that the DCSBM is a special DCMM model when all $\pi_i$'s are degenerate (i.e., has only one nonzero entry which is equal to 1, and the other entries are zero. The corresponding node is also called as pure node), and further, when all $\theta_i$'s are equal to each other (i.e., no degree heterogeneity), the DCMM model degenerates to the general SBM. In addition, the mixed membership stochastic block model (MMSBM), proposed by \cite{Airoldi:2008}, is also a special case when $\theta_i$'s are equal to each other but $\pi_i$'s are non-degenerate). The research interest in the DCMM model mainly focuses on estimating the membership matrix $\Pi$. \cite{Jin:2023} proposed a simplex-based method. They found that each row of the SCORE normalized adjacency matrix falls in a simplex, and the simplex depends on the membership matrix. By SCORE normalizing the Laplacian matrix, \cite{Ke:2022} used the simplex-based method to estimate the membership vectors under the severe degree heterogeneity, respectively. Under the setting of no degree heterogeneity, \cite{Mao:2021} proposed a fast and provably consistent algorithm, called ``sequential projection after cleaning (SPACL)", to estimate the membership matrix. It is worth noting that the current inference methods for DCMM models are based on the simplex structure, and have been a largely under-explored domain, especially in estimating the number of communities $K$. 

\textit{\textbf{Latent space model.}} The latent space model (LSM), proposed by \cite{Hoff:2002}, is also a widely used network model. The LSM assumes that each node is mapped to a latent position $\bx_i\in\bR^d$. Conditionally on the collection of latent positions $\bm{X} = [\bx_1,\ldots,\bx_n]^\top$, the edge between nodes $i$ and $j$ is Bernoulli random variable with probability $p_{ij}=\kappa(\bx_i,\bx_j)$, where $\kappa(\cdot,\cdot)$ is a symmetric kernel function. The two most commonly used kernel functions are inner product functions $\kappa(\bx,\bm{y})=\bx^\top\bm{y}$ and generalized inner product functions $\kappa(\bx,\bm{y})=\bx^\top I_{a,b}\bm{y}$, where $a+b=d$ and $I_{a,b}$ for integers $a \geq 1$ and $b\geq0$ is a diagonal matrix with $a$ ``1" followed by $b$ ``$-1$". Then, these two kernel functions correspond to the random dot product graph (RDPG) \citep{Nickel:2008} and its generalised version (GRDPG) \citep{Rubin:2022}, respectively. To estimate the latent positions, \cite{Sussman:2014} proposed an adjacency spectral embedding (ASE) method using the eigenvectors associated with the top $d$ eigenvalues of the adjacency matrix. However, \cite{Xie:2020} pointed out that the ASE method formulates the problem in a low-rank matrix factorization manner, but it neglects the Bernoulli likelihood information present in the sampling model. Hence, \cite{Xie:2023} proposed an effective one-step procedure to estimate the latent positions. In addition, the issue of the hypothesis test has received considerable attention, that is, determining whether or not two nodes $i$ and $j$ in an LSM have the same latent positions \citep{Du:2023}. Under the framework of DCMM model, let $\bx_i=\sum_{k}\pi_i(k)v_i$ for $i=1,\ldots,n$ by choosing some $v_1,\ldots,v_K\in\bR^d$ for some $d = \mathrm{rank}(B) \leq K$ such that $v_k^\top I_{a,b}v_l=B_{kl}$, for all $k,l\in\{1,\ldots,K\}$ where $a$ is the number of positive eigenvalues of $B$ and $b=d-a$. Then, the LSM degenerates to the MMSBM.

Here, we have introduced some commonly used network models. It is not difficult to see that the difference in different models is that the probability matrices $P$ have different structures, and there is an inclusion relationship between the different models. Hence, in the network data analysis, the core problem is fitting the network to an appropriate model and estimating the corresponding parameters. Given a random network $\mathcal{G}$, identifying which model is suitable for a network is an interesting research issue. Intuitively, if one fits the network data to an incorrect model, then we can not correctly infer the statistical properties of the network. In this article, we first establish a spectral property of the normalized adjacency matrix and provide a goodness-of-fit test algorithm of models.

\subsection{A spectral-based statistic}

In the network analysis, hypothesis testing mainly focuses on the SBM and its variants, especially in testing the structure of communities. For an adjacency matrix $A$ of SBM, the normalized adjacency matrix $\bar{A}$ is defined as follows:
\[
\bar{A}_{ij}=
\begin{dcases}
	\dfrac{A_{ij}-p_{ij}}{\sqrt{p_{ij}(1-p_{ij})}} & i\neq j, \\
	0 & i=j,
\end{dcases}
\]
where $p_{ij}=\bE\{A_{ij}\}$. The majority of statistics are based on this normalized adjacency matrix. \cite{Lei:2016} showed that the extreme eigenvalues of the matrix $(n-1)^{-1/2}\bar{A}$ asymptotically follows the Tracy-Widom distribution with index 1. Similarly, \cite{Wu:2024} showed the trace of the matrix $(n^{-1/2}\bar{A})^3$  asymptotically follows the normal distribution. Correspondingly, the empirically normalized adjacency matrix, i.e., the $p_{ij}$'s are replaced by its estimates $\hat{p}_{ij}$'s, also have identical limiting distribution. Under these results, they implement the test $H_0:K=K_0$ under the framework of SBM. Further, by sequential testing, the number of communities can be estimated.

It is not hard to see that the basic idea is to use an accuracy probability matrix estimator to normalize the adjacency matrix. Then, the corresponding limiting properties are established. However, the existing method mainly focused on a given model. Here, we consider extending the results to the network from the general model.

Naturally, for an adjacency matrix $A$ from network $\mathcal{G}$, the normalized adjacency matrix is
\begin{equation}\label{eq:spec1}
	\tilde{A}^*_{ij} = \begin{dcases}\dfrac{A_{ij}-p_{ij}}{\sqrt{np_{ij}(1-p_{ij})}}, & i\neq j,\\ 0, & i=j,\end{dcases}
\end{equation}
where $p_{ij}=\bE\{A_{ij}\}$ for all $1\leq i\neq j\leq n$. Then $\tilde{A}^*$ is a generalized Wigner matrix satisfying $\bE(\tilde{A}^*_{ij})=0$ and $\var(\tilde{A}^*_{ij})=1/n$ for all $1\leq i\neq j\leq n$. Combining results in \cite{Bai:2016} and \cite{Wang:2021} we have
\begin{equation}\label{eq:limit1}
	\dfrac{1}{\sqrt{6}}\tr((\tilde{A}^{*})^3)\rightsquigarrow N(0,1).
\end{equation}
We formally state and prove this result as Lemma \ref{lemma:1} in the Appendix.

Notice that the matrix $\tilde{A}^*$ involves unknown parameters $p_{ij}$'s. Hence, we can consider a natural estimate of $\tilde{A}^*$ by plugging in the estimated parameters. Let $\hat{p}_{ij}$ be an estimate of $p_{ij}$. Then, the estimates $\hat{p}_{ij}$'s lead to the empirically normalized adjacency matrix $\tilde{A}$:
\begin{equation}\label{eq:spec2}
	\tilde{A}_{ij} = \begin{dcases}\dfrac{A_{ij}-\hat{p}_{ij}}{\sqrt{n\hat{p}_{ij}(1-\hat{p}_{ij})}}, & i\neq j,\\ 0, & i=j.\end{dcases}
\end{equation}
It is natural to conjecture that when the estimates $\hat{p}_{ij}$'s are accurate enough, the convergence in \eqref{eq:limit1} will still hold for $\tilde{A}$. To obtain the asymptotic result of $\tilde{A}$, we first make the following assumptions:

\begin{assumption}\label{ass:1}
Let $\hat{p}_{ij}$ be the estimator of $p_{ij}$ for all $1\leq i,j\leq n$. Denote matrix $\Delta^{\prime}=[\Delta^{\prime}_{ij}]_{n\times n}$, where $\Delta_{ij}^{\prime}=\dfrac{p_{ij}-\hat{p}_{ij}}{\sqrt{n p_{ij}(1-p_{ij})}}$ for $i\neq j$ and $\Delta_{ii}^{\prime}=0$. The difference between $p_{ij}$ and $\hat{p}_{ij}$ satisfies
	\begin{enumerate}
		\item[(i)] $\max_{ij}\vert\hat{p}_{ij}-p_{ij}\vert=o_p(n^{-1/4})$; 
		\item[(ii)] $\tr((\Delta^{\prime})^3)=o_p(1)$.
	\end{enumerate}
\end{assumption}

Assumption \ref{ass:1} gives some restrictions for the estimators $\hat{p}_{ij}$'s. In statistical learning, the model parameter can be accurately estimated based on an appropriate model, and poor models will lead to significant deviations in the estimator of corresponding parameters. Since the results are established on the general network model, and the true model is not specified, we only require the estimators $\hat{p}_{ij}$'s to be accurate enough. These conditions are extremely mild. For example, under SBMs with balanced community structure and the true number of communities, the standard large deviation inequality suggests the $\max_{kl}|B_{kl}-\hat{B}_{kl}|=o_p(K\log n/n)$, which implies $\max_{ij}|p_{ij}-\hat{p}_{ij}|=o_p(K\log n/n)$. Further, we also have $\tr((\Delta^{\prime})^3)=o_p(K^3n^{-3/2}\log^3n)$. Hence, as long as $K=O(\sqrt{n}/\log n)$, the conditions hold under the framework of the SBM. For the $\beta$-model, \cite{Chatterjee:2011} shows that, if $L_n =o(\log(\log n))$, then $\max_i|\hat{\beta}_{i}-\beta_{i}|=O_p(n^{-1/2}\log^{-1/2}n)$, which implies $\max_{ij}|\hat{p}_{ij}-p_{ij}|=O_p(n^{-1/2}\log^{-1/2}n)$. However, it is difficult to verify the conditions (iii) since the technical and complex dependency among $\hat{p}_{ij}$'s. By simulation study, we set $\beta_i=iL_n/n$ for all $1\leq i\leq j$, and Table \ref{tab:trace} shows that the values of $\tr((\Delta^{\prime})^3)$ under the different settings. As shown in Table \ref{tab:trace}, the values of $\tr((\Delta^{\prime})^3)$ are smaller and smaller with the sample size increasing. Hence, we can assert that $\tr((\Delta^{\prime})^3)$ tends to 0.

\begin{table}[h]
\setlength{\abovecaptionskip}{0cm}  
\setlength{\belowcaptionskip}{0.5cm} 
\centering
\caption{The values of $\tr((\Delta^{\prime})^3)$ under the $\beta$-model.}
\label{tab:trace}
\begin{tabular*}{\textwidth}{c@{\extracolsep{\fill}}ccccc}
\toprule
& $L_n=0$ & $L_n=(\log(\log n))^{1/3}$ & $L_n=\log(\log n)$ &  $L_n=(\log n)^{1/2}$ \\ \midrule
$n=200$ & $2\times10^{-4}$ & $-4\times10^{-3}$ & $-3\times10^{-3}$ & $-9\times10^{-3}$ \\
$n=600$ & $3\times10^{-5}$ & $-9\times10^{-5}$ & $-3\times10^{-4}$ & $-1\times10^{-3}$ \\
$n=1000$ & $-2\times10^{-5}$ & $-5\times10^{-5}$ & $-2\times10^{-4}$ & $-8\times10^{-4}$ \\\bottomrule
\end{tabular*}
\end{table}

Formally, we  give the following theorem:

\begin{theorem}\label{thm:null}
	Let $A$ be an adjacency matrix generated from a network model. Let $\tilde{A}$ be given as in \eqref{eq:spec2} using estimators $\hat{p}_{ij}$'s. Suppose that Assumption \ref{ass:1} holds. Then, we have the following result:
	\begin{equation}\label{eq:limit2}
		T_n:=\dfrac{1}{\sqrt{6}}\tr(\tilde{A}^{3})\rightsquigarrow N(0,1),
	\end{equation}
where ``$\rightsquigarrow$" denotes convergence in distribution.
\end{theorem}

\textbf{Remark 1.} Theorem \ref{thm:null} is proved in Appendix. Theorem \ref{thm:null} indicates that as long as the accuracy of estimators $p_{ij}$'s satisfy a mild condition, the trace of the third-order for the empirically normalized adjacency matrix will convergences to a normal distribution. This theorem is also a nontrivial generalization of Theorem 1 in \cite{Dong:2020} and Theorem 2 in \cite{Wu:2024}.

Using this result, we can consider implementing the goodness-of-fit of the network model. In statistical learning, for a network $A$, it is significant to determine an appropriate model to fit this network. According to Theorem \ref{thm:null}, we know that if we can obtain enough accurate estimates $\hat{p}_{ij}$'s, then the statistic $T_n$ convergences in distribution to the standard normal distribution. Specifically, we assume that the network is generated from the model $M_1$ with parameter $\Theta$. For example, we assume that the model $M_1$ is the SBM, and the parameter $\Theta = (K,\sigma,B)$. Then, a hypothesis test problem can be considered as follows:
\begin{equation}\label{eq:test}
	H_0: A\ \text{is generated from the model $M_1$}\ v.s.\ H_1: A\ \text{is generated from other models.}
\end{equation}
Then, based on the trace of the third-order for a normalized adjacency matrix, we propose a spectral statistic to test the hypothesis \eqref{eq:test}. First, based on the model $M_1$, we use a sample $A$ to estimate the parameter $\hat{\Theta}$, and obtain the estimate $\hat{P}$ of the link-probability matrix. Second, we use the estimate $\hat{P}$ to compute the empirically normalized adjacency matrix $\tilde{A}$, and obtain $T_n=\tr((\tilde{A})^3)/\sqrt{6}$. According to Theorem \ref{thm:null}, if the network $A$ is generated from the model $M_1$, the asymptotic distribution of the statistic $T_n$ is the standard normal distribution. Under the alternative hypothesis, however, inappropriate models will lead to low accuracy in parameter estimation. Moreover, the adjacency cannot be correctly normalized, which will lead to a large deviation by the normalization term. We perform simulation studies and find that the empirical distribution of $T_n$ does not deviate from the standard normal distribution under the null hypothesis. Using the above results, we can carry out the hypothesis test. Then, we have a rejection rule:
\[
\text{Reject}\ H_0,\ \text{if}\ |T_n|\geq u_{1-\alpha/2},
\]
where $u_{1-\alpha/2}$ is the upper $\alpha$-th quantile of the standard normal distribution. The corresponding hypothesis test algorithm can be seen in Algorithm \ref{algo:GoF}.

\begin{algorithm}[htbp] 
\caption{Goodness-of-fit for the network model.} 
\label{algo:GoF} 
\begin{algorithmic}[1] 
\REQUIRE Adjacency matrix $A$, the candidate model $M_1$, and the nominal level $\alpha$.
\STATE Based on the candidate model $M_1$, using the network sample $A$ to estimate the model parameter $\hat{\Theta}$.
\STATE Using $\hat{\Theta}$ to calculate $\hat{p}_{ij}$ for all $1\leq i,j\leq n$ under the framework of model $M_1$. And, compute $\tilde{A}$ using \eqref{eq:spec2}.
\STATE Calculate $T_n = \dfrac{1}{\sqrt{6}}\tr((\tilde{A})^3)$ and $p_{value} = 2\bP_{N(0,1)}\{X > |T_n|\}$.
\IF{$p_{value} > \alpha$} 
\STATE One asserts that $A$ is generated from the model $M_1$. 
\ELSE 
\STATE One asserts that $A$ is not generated from the model $M_1$.
\ENDIF  
\end{algorithmic} 
\end{algorithm}

\textbf{Remark 2.} Specifically, Algorithm \ref{algo:GoF} can be used to test the node homogeneous for $\beta$-model. In the $\beta$-model with $n$ nodes, one of the interest problems is the node homogeneous, i.e., $\beta_1=\beta_2=\cdots=\beta_r$ where $1\leq r\leq n$. For the proposed method, we can consider setting the candidate model as the $\beta$-model with $\beta_1=\cdots=\beta_r\neq\beta_{r+1}\neq\cdots\neq\beta_n$. In addition, under the homogeneous assumption with $r=n$, the $\beta$-model reduces to the E-R model. Hence, one can set the candidate model as the E-R model. Simulation shows that the proposed testing method can test the homogeneous null hypothesis.

\section{Model selection in degree-corrected mixed membership models}\label{sec:appdcmm}

In this section, we apply the proposed method to DCMM models for estimating the number of communities. As discussed in Section \ref{sec:model}, the difference between the DCMM model and DCSBM is whether the membership vector is degenerate. In the research about DCMM models, the number of communities is known. In practice, however, the number of communities is usually unknown. Hence, accurately estimating the number of communities is of great practical and theoretical significance. To the best of our knowledge, too little work is devoted to determining the number of communities in DCMM models. Compared with SBMs and DCSBMs, the prior information of the membership vector of a node is more complicated in DCMM models, which makes it difficult to use the method based on information criterion for DCMM models. Similar to \cite{Lei:2016}, based on the proposed goodness-of-fit test, we consider a sequential testing method that can be suitable for DCMM models.

In Section \ref{sec:method}, we gave a general theory for the goodness-of-fit test of network models. This method assumes that the network is generated from a candidate model $M_1$. Specifically, to estimate the number of communities, let the candidate model $M_1$ be the DCMM model with $K_0$ communities. Hence, let the network $A$ be generated by a DCMM model with $K$ communities, the hypothesis test problem can be concretized as
\begin{equation}\label{eq:testDCMM}
	H_{0,K_0}: K=K_0\ v.s.\ H_{1,K_0}: K\neq K_0,
\end{equation}
where $K$ and $K_0$ denote true and a hypothetical number of communities for DCMM models, respectively. Under this setting, we can use the method in Algorithm \ref{algo:GoF} to calculate the statistic $T_n(K_0)$. Then, if $|T_n(K_0)|>u_{1-\alpha/2}$ for a nominal level $\alpha$, we reject the null hypothesis $H_{0,K_0}$. Following this idea, given a maximum value $K_{max}$, we can compute the statistic sequence $T_{n,1},\cdots,T_{n,K_{max}}$, and $T_{n,K_0}$ should be large than $u_{1-\alpha/2}$ when $K_0\neq K$. Hence, for a given $\alpha$, the estimated number of communities is given by 
\begin{equation}\label{eq:EstK}
	\hat{K} = \min\{K_0\in\{1,\cdots,K_{max}\}: \left|T_n(K_0)\right| < u_{1-\alpha/2}\}.
\end{equation}

In DCMM models, the non-identifiability of the model is an intrinsic issue. For fixed $K$, \cite{Jin:2023} studied the identifiability, and showed that the model is identifiable when the probability matrix $B$ has unit diagonals and each community has at least one pure node, i.e., for eligible $(\Theta_1,\Pi_1,B_1)$ and $(\Theta_2,\Pi_2,B_2)$, if $\Theta_1\Pi_1B_1\Pi_1^\top\Theta_1=\Theta_2\Pi_2B_2\Pi_2^\top\Theta_2$, we have $\Theta_1=\Theta_2,\Pi_1=\Pi_2$, and $B_1=B_2$. In fact, there is a major concern for the identifiability of $K$. Let a DCMM model with $K$ communities has structure $\Theta\Pi B\Pi^\top\Theta$. However, there may exist $\tilde{K}\neq K$ and $(\tilde{\Theta},\tilde{\Pi},\tilde{B})$ such that $\tilde{\Theta}\tilde{\Pi}\tilde{B}\tilde{\Pi}^\top\tilde{\Theta}=\Theta\Pi B\Pi^\top\Theta$. Define $\mathbb{Q}(k,l)$ as a class of $k\times l$ matrix that satisfies that the sum of elements in each column is one, that is, $\mathbb{Q}(k,l) = \{Q\in\bR^{k\times l}:Q^\top\bm{1}_k=\bm{1}_l\}$. We also define $\mathbb{S}_m=\{\bm{x}\in\mathbb{R}^{m}:\sum_{i=1}^m x_i=1\}$. Then, we have the following proportion:

\begin{proposition}\label{prop:idendcmm}
Let $\bm{y}\in\mathbb{S}_l$ be a vector. Then, for any matrix $Q\in\mathbb{Q}(k,l)$, we have $\bm{x}=Q\bm{y}\in\mathbb{S}_k$.
\end{proposition}

\begin{proof}
	Let $Q=(q_{ij})_{k\times l}$. It is easy to have
	\begin{equation*}
		\begin{dcases}
		x_1 = q_{11}y_1 + q_{12}y_2 + \cdots + q_{1l}y_l, \\
		x_2 = q_{21}y_1 + q_{22}y_2 + \cdots + q_{2l}y_l, \\
		\quad \vdots \\
		x_k = q_{k1}y_1 + q_{k2}y_2 + \cdots + q_{kl}y_l. \\
		\end{dcases}
	\end{equation*}
	Notice that, since $ \sum_{i=1}^l y_i  = 1$ and $q_{1j}+q_{2j}+\cdots+q_{kj} = 1$ for all $1\leq j\leq l$, then we have $\sum_{i=1}^k x_i = 1$. Thus, $\bm{x}\in\mathbb{S}_k$.
\end{proof}

Proposition \ref{prop:idendcmm} indicates that a membership vector with $l$ communities can be transformed into a membership vector with $k$ communities. For a $\tilde{K}\times K$ matrix $Q\in\mathbb{Q}(\tilde{K},K)$, let $\tilde{\Pi}=\Pi Q^\top\in[0,1]^{n\times\tilde{K}}$. Hence, as long as $Q^\top\tilde{B}Q=B$, we have $\Theta\tilde{\Pi}\tilde{B}\tilde{\Pi}^\top\Theta=\Theta\Pi B\Pi^\top\Theta$. Hence, the non-identifiable of $K$ results in that we need to consider the two-sided alternative for the hypothesis test \eqref{eq:testDCMM}. Due to the non-identifiability of $K$, it is not easy to prove the consistency of estimation \eqref{eq:EstK}, and the method may tend to underestimate the number of communities. In the simulation, we empirically investigate the performance of estimation, and the results show that the proposed sequential testing method can exactly estimate the number of communities in most cases.

\section{Simulation}\label{sec:simulation}
In this section, we verify the effectiveness of the proposed method through extensive simulation studies. In the $\beta$-model setting, the MLE in \cite{Chatterjee:2011} is used to estimate the parameters $\beta_i$. A disadvantage is that this estimation procedure will not work when the network is sparse. In the SBM and DCSBM settings, we apply the corrected Bayesian information criterion (CBIC) proposed by \cite{Hu:2020} and spectral clustering on ratios-of-eigenvectors (SCORE) proposed by \cite{Jin:2015} to estimate the number of communities and the community label. In the LSM settings, the adjacency spectral embedding method, proposed by \cite{Sussman:2014}, is used to estimate the latent position. All simulations were performed on a PC with a single processor of 2.3 GHz 8‐Core Intel Core i9.

\subsection{The null distribution}\label{sec:null}

In the simulation, we examine the finite sample null distribution of the proposed test statistic. Here, we consider the five basic network models: E-R model, $\beta$-model, stochastic block model, degree-corrected stochastic block model, and latent space model. For all models, we set $n=500$ and $1000$. Other parameter settings for different models are as follows:

(1) \textit{E-R model}: Let $p = 0.01, 0.05$, and $0.1$;

(2) \textit{$\beta$-model}: Let $\beta_i=iL_n/n$, where $L_n=0,(\log(\log n))^{1/2}, (\log n)^{1/2}$; 

(3) \textit{Stochastic block model}: Let the number of communities as $K=3$ and the edge probability between communities $u$ and $v$ as $B_{uv}=\rho(1+4\times I[u=v])$, where $\rho$ measures the sparsity of network. The community label $\sigma_i$'s are drawn independently from the multinomial distribution with parameter $\pi=(1/3,1/3,1/3)^\top$. We consider the cases of $\rho = 0.02, 0.05$, and $0.1$;

(4) \textit{Degree-corrected stochastic block model}: The community label $\sigma$ and probability matrix $B$ are generated the same way as for the stochastic block model. In addition, following the method in \cite{Zhao:2012}, we generate the degree-corrected parameters. Specifically, 
\[
\theta_i=\begin{dcases}u_i & \text{w.p.}\ 0.8, \\ 9/11 & \text{w.p.}\ 0.1, \\ 13/11 & \text{w.p.}\ 0.1,\end{dcases}
\]
where $u_i\sim\mathrm{Unif}[4/5,6/5]$. We also consider the cases of $\rho = 0.02, 0.05$, and $0.1$;

(5) \textit{Latent space model}: For the latent space model, we consider the case of the random dot product graph with latent dimension $d=1$. The latent position $\bm{x}_{0i}$ for the $i$th node is set to $\bm{x}_{0i}=0.8\cdot\sin\{\pi(i-1)/(n-1)\}+0.1$, where $1\leq i\leq n$.  Let $\bm{X}=\rho\bm{X}_0 = \rho[\bm{x}_{01},\ldots,\bm{x}_{0n}]^\top$, where $\rho = 0.2, 0.5$ and $1$.

We plot the normal Q-Q plot of the statistic from 1000 data replications. Figures \ref{fig:qq-er} - \ref{fig:qq-lsm} show the results for the Q-Q plot under the different null models. It is easy to see that for the different null models, the statistic $T_n$ convergences in distribution to the standard normal distribution. Compared with other test methods of network data, such as the largest singular value \citep{Lei:2016} and the maximum entry-wise deviation \citep{Hu:2021}, the proposed test statistic is not necessary to consider the bootstrap correction, and improve the test efficiency. The results visually demonstrate the results in Theorem \ref{thm:null}.

\begin{figure}[htbp]
	\centering
	\vspace{-0.35cm}
	\setlength{\abovecaptionskip}{-2pt}
	\subfigtopskip=2pt
	\subfigbottomskip=2pt
	\subfigcapskip=-5pt
	\subfigure[$p=0.01$]{\label{fig:er-sub1.1}
	\includegraphics[width=0.23\linewidth]{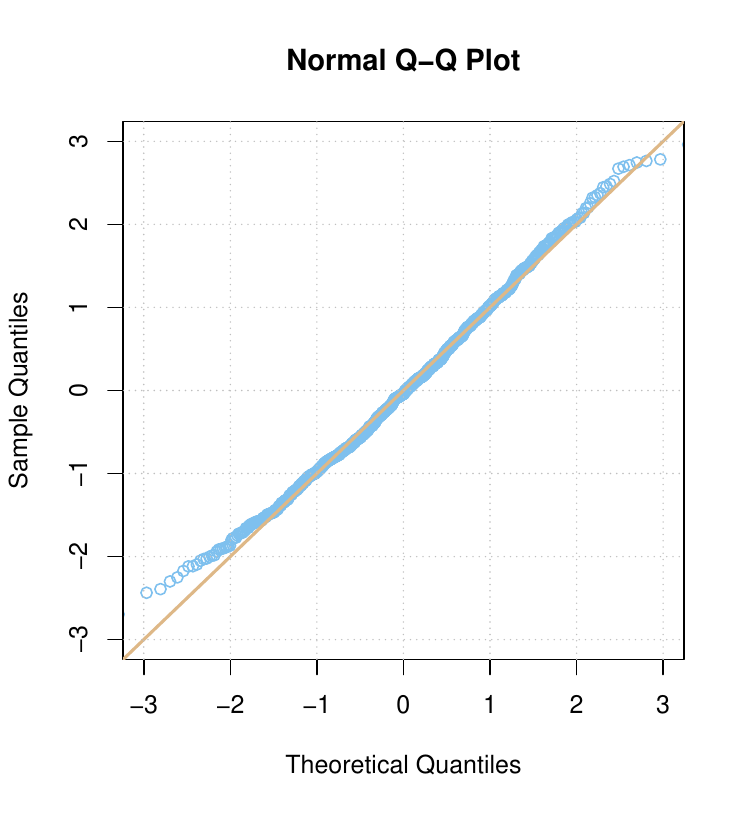}}
	\subfigure[$p=0.05$]{\label{fig:er-sub1.2}
	\includegraphics[width=0.23\linewidth]{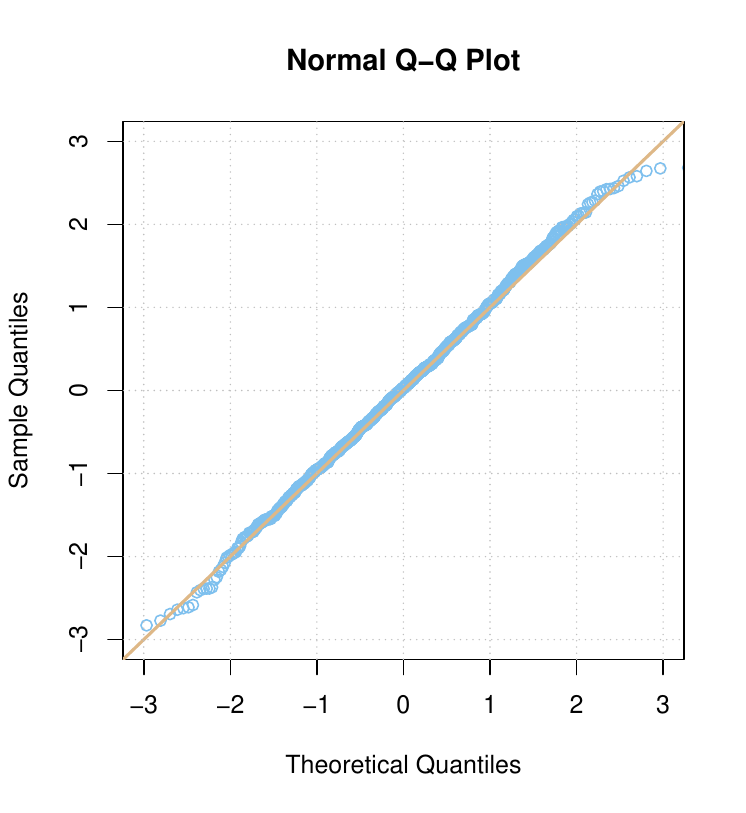}}
	\subfigure[$p=0.1$]{\label{fig:er-sub1.3}
	\includegraphics[width=0.23\linewidth]{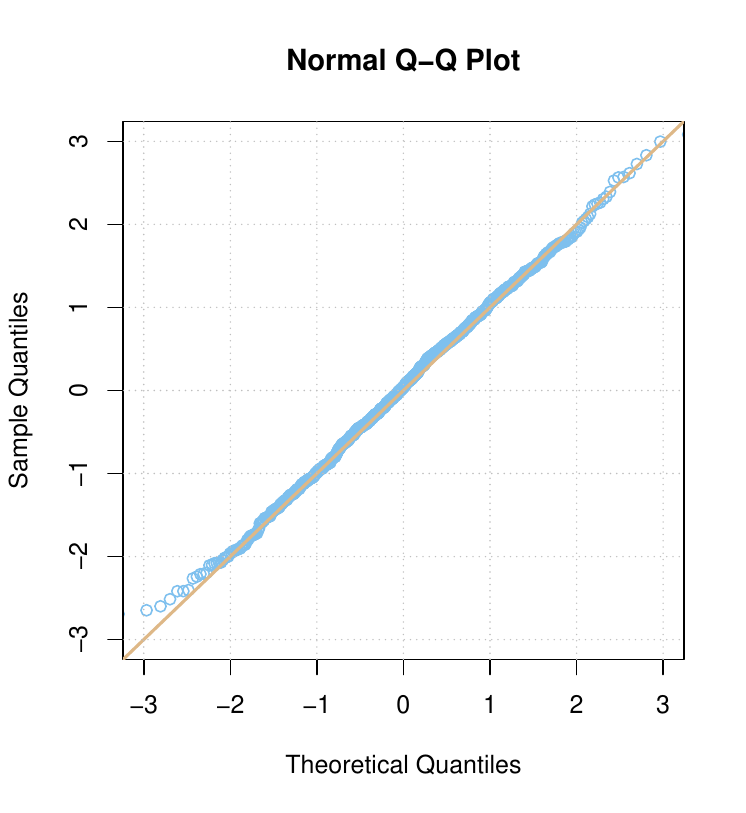}}
	
	\subfigure[$p=0.01$]{\label{fig:er-sub2.1}
	\includegraphics[width=0.23\linewidth]{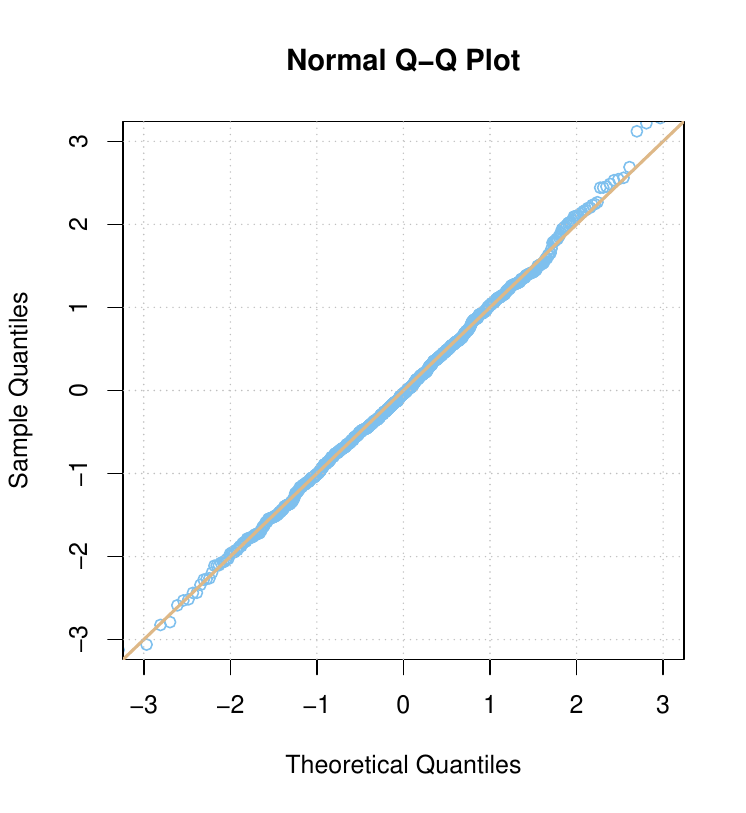}}
	\subfigure[$p=0.05$]{\label{fig:er-sub2.2}
	\includegraphics[width=0.23\linewidth]{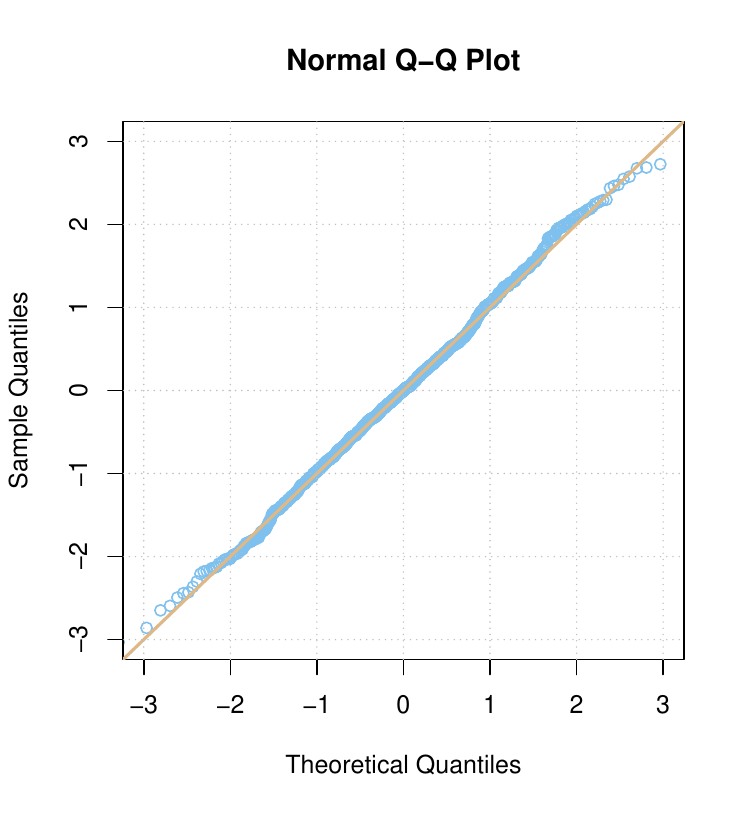}}
	\subfigure[$p=0.1$]{\label{fig:er-sub2.3}
	\includegraphics[width=0.23\linewidth]{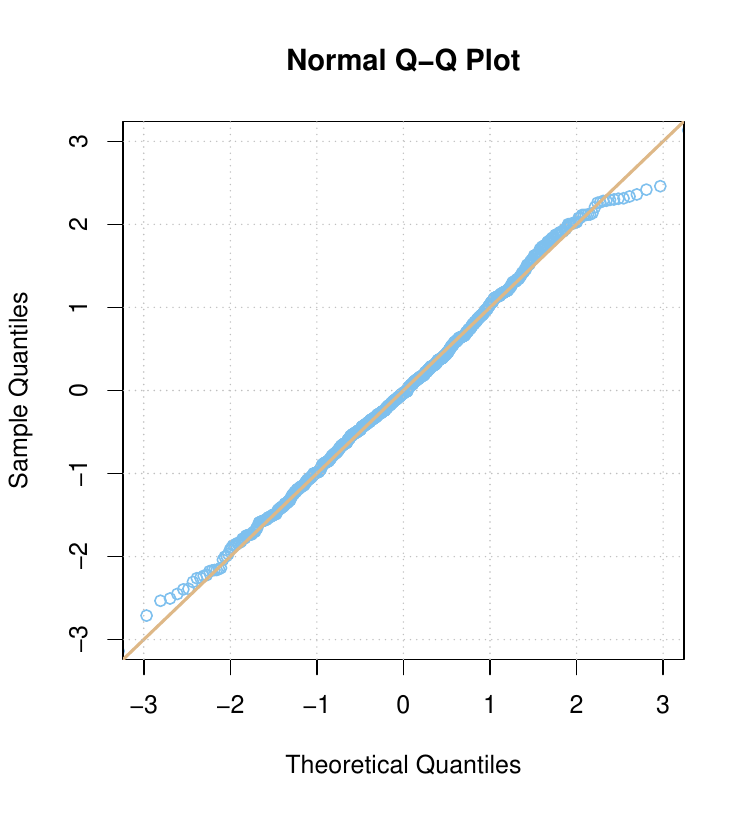}}
	\vspace{0.5cm}
	\caption{Normal Q-Q plot under the E-R model when $n=500$ (upper row) and $n=1000$ (lower row).}
	\label{fig:qq-er}
\end{figure}

\begin{figure}[htbp]
	\centering
	\vspace{-0.35cm}
	\setlength{\abovecaptionskip}{-2pt}
	\subfigtopskip=2pt
	\subfigbottomskip=2pt
	\subfigcapskip=-5pt
	\subfigure[$L_n=0$]{\label{fig:beta-sub1.1}
	\includegraphics[width=0.23\linewidth]{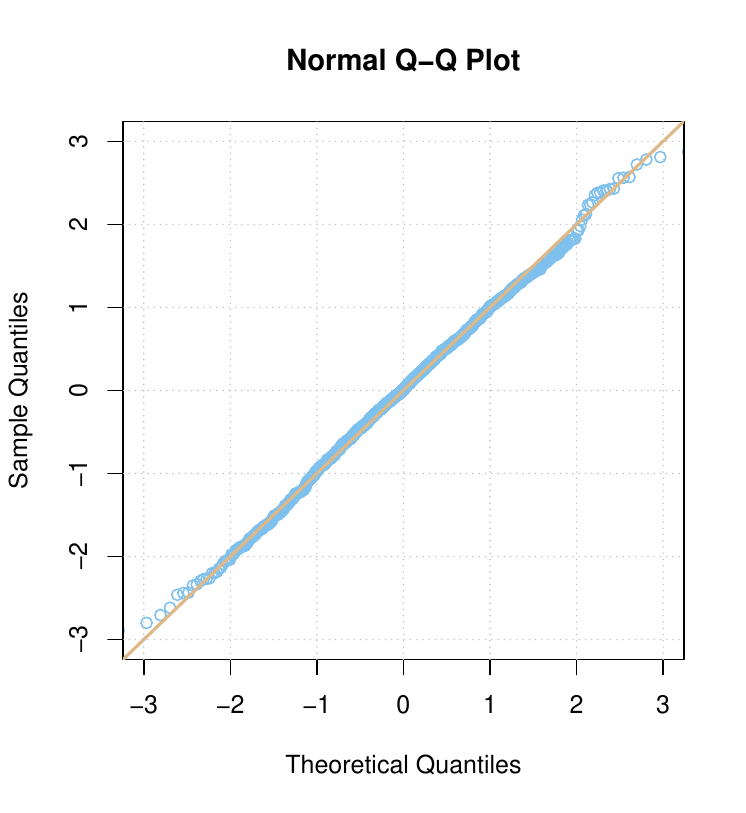}}
	\subfigure[$L_n=(\log(\log n))^{1/2}$]{\label{fig:beta-sub1.2}
	\includegraphics[width=0.23\linewidth]{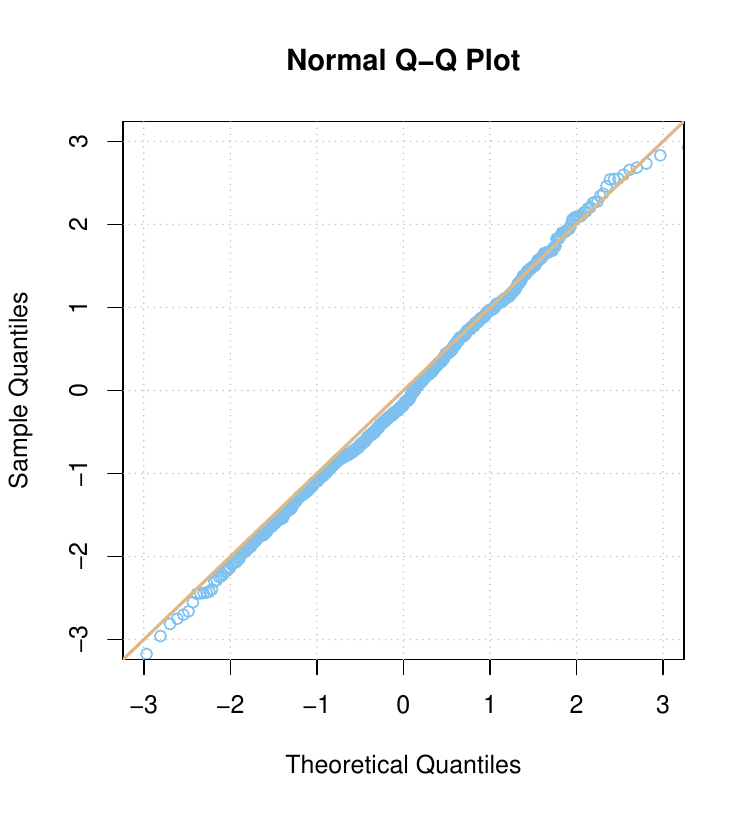}}
	\subfigure[$L_n=(\log n)^{1/2}$]{\label{fig:beta-sub1.4}
	\includegraphics[width=0.23\linewidth]{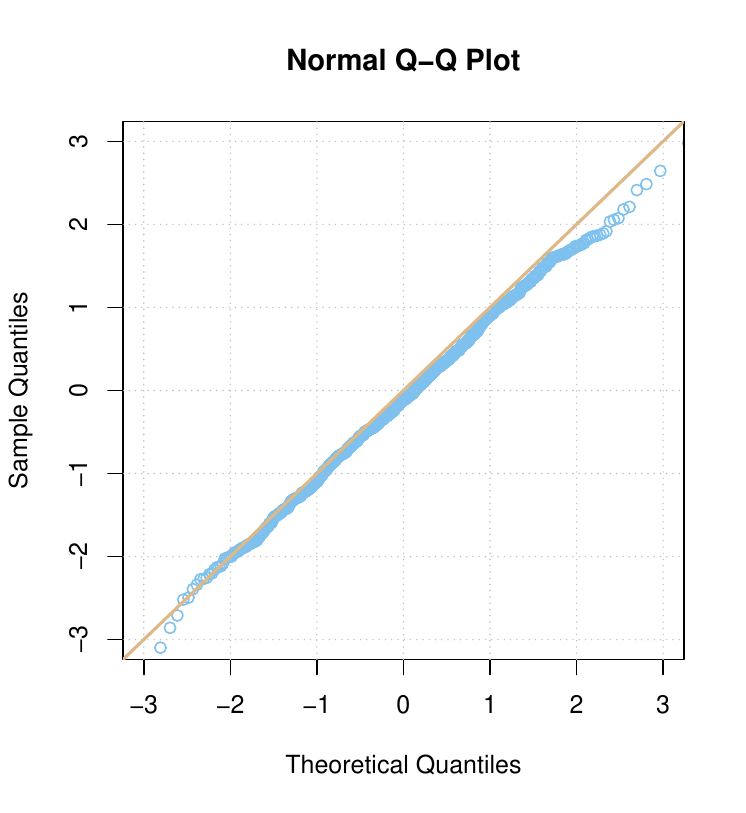}}
	
	\subfigure[$L_n=0$]{\label{fig:beta-sub2.1}
	\includegraphics[width=0.23\linewidth]{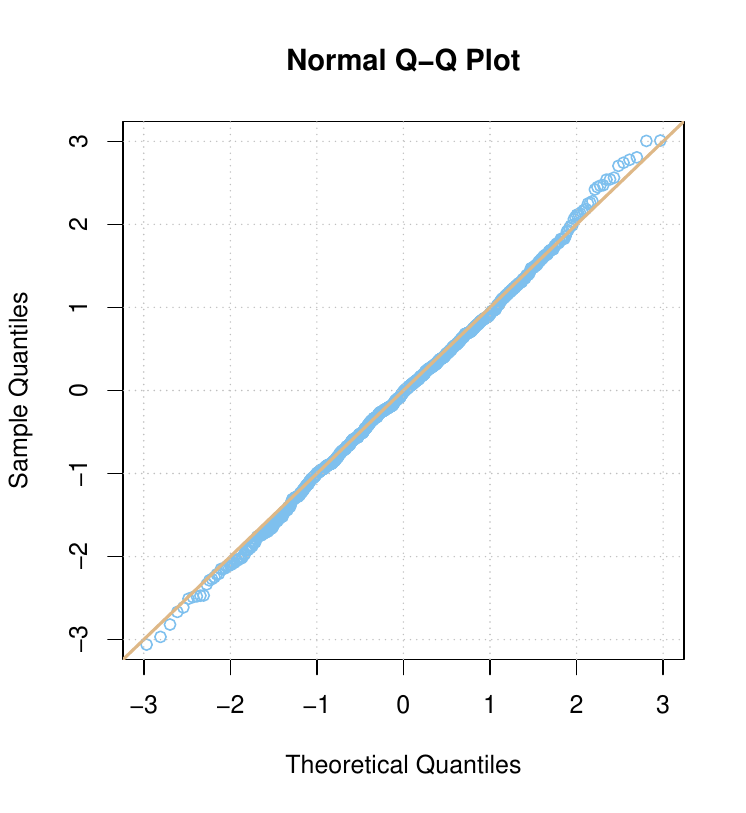}}
	\subfigure[$L_n=(\log(\log n))^{1/2}$]{\label{fig:beta-sub2.2}
	\includegraphics[width=0.23\linewidth]{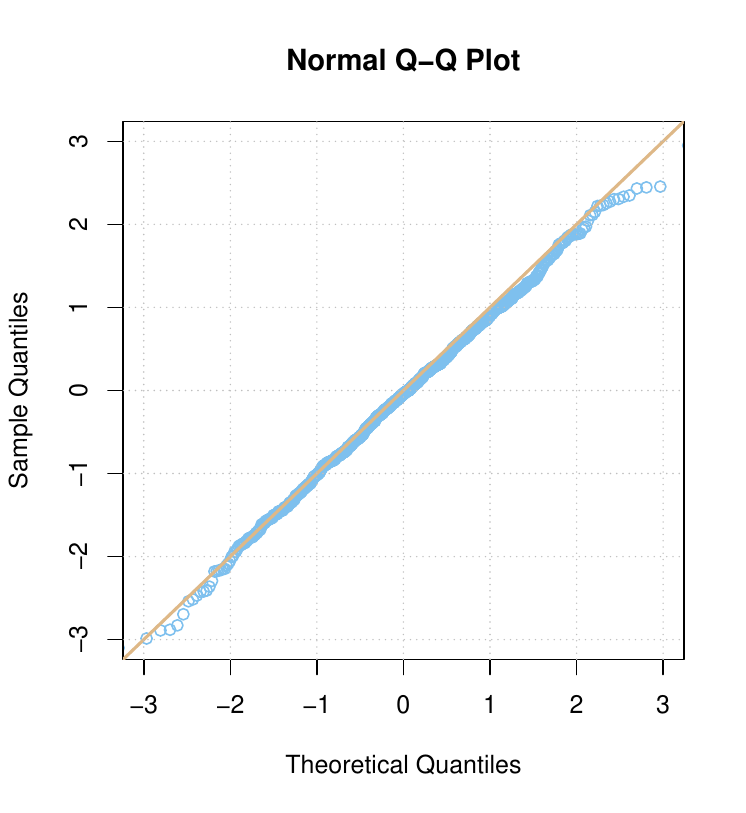}}
	\subfigure[$L_n=(\log n)^{1/2}$]{\label{fig:beta-sub2.4}
	\includegraphics[width=0.23\linewidth]{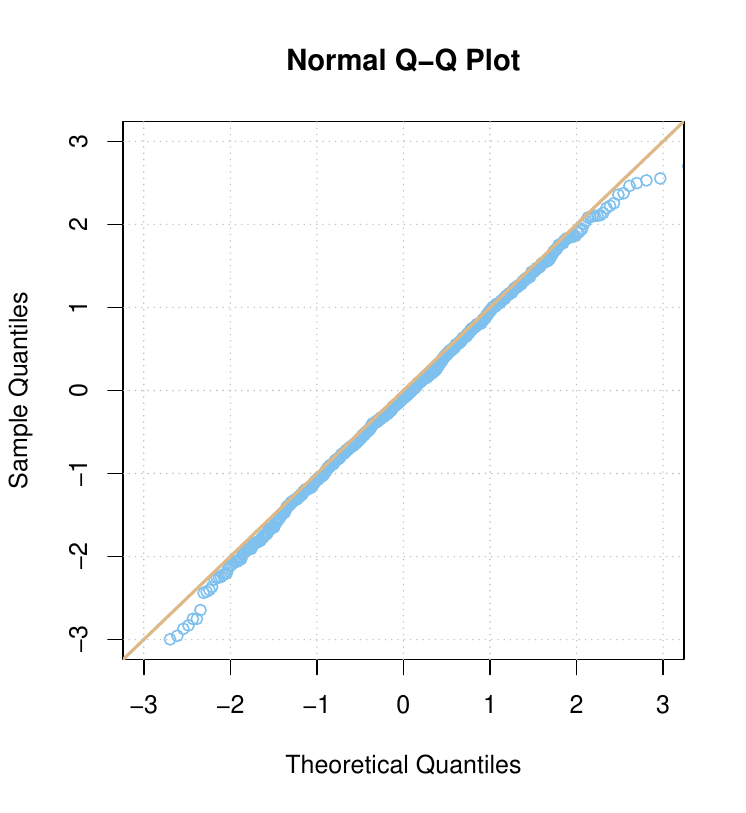}}
	\vspace{0.5cm}
	\caption{Normal Q-Q plot under the $\beta$-model when $n=500$ (upper row) and $n=1000$ (lower row).}
	\label{fig:qq-beta}
\end{figure}

\begin{figure}[htbp]
	\centering
	\vspace{-0.35cm}
	\setlength{\abovecaptionskip}{-2pt}
	\subfigtopskip=2pt
	\subfigbottomskip=2pt
	\subfigcapskip=-5pt
	\subfigure[$\rho=0.02$]{\label{fig:sbm-sub1.1}
	\includegraphics[width=0.23\linewidth]{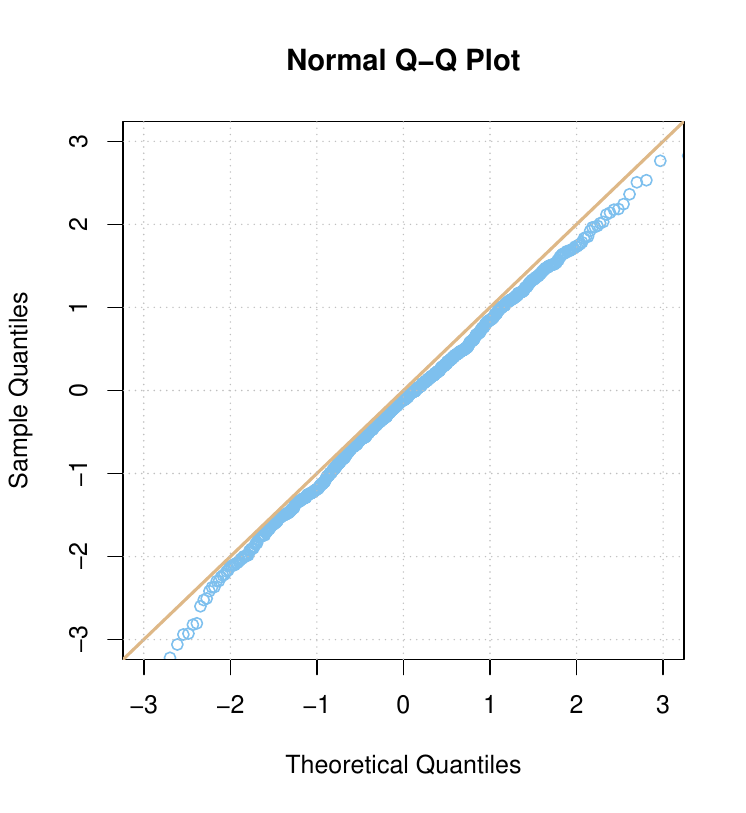}}
	\subfigure[$\rho=0.05$]{\label{fig:sbm-sub1.2}
	\includegraphics[width=0.23\linewidth]{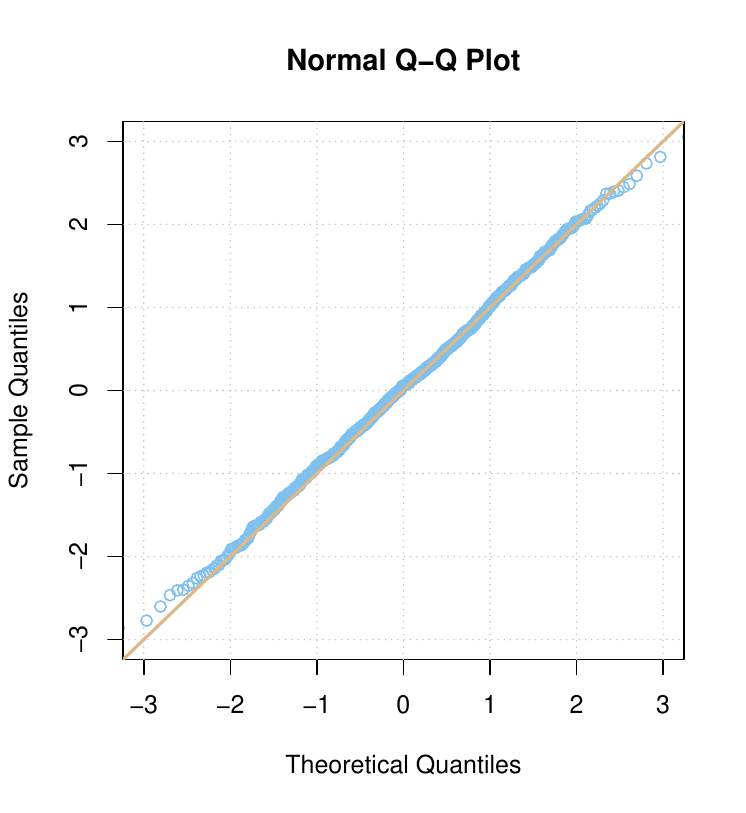}}
	\subfigure[$\rho=0.1$]{\label{fig:sbm-sub1.3}
	\includegraphics[width=0.23\linewidth]{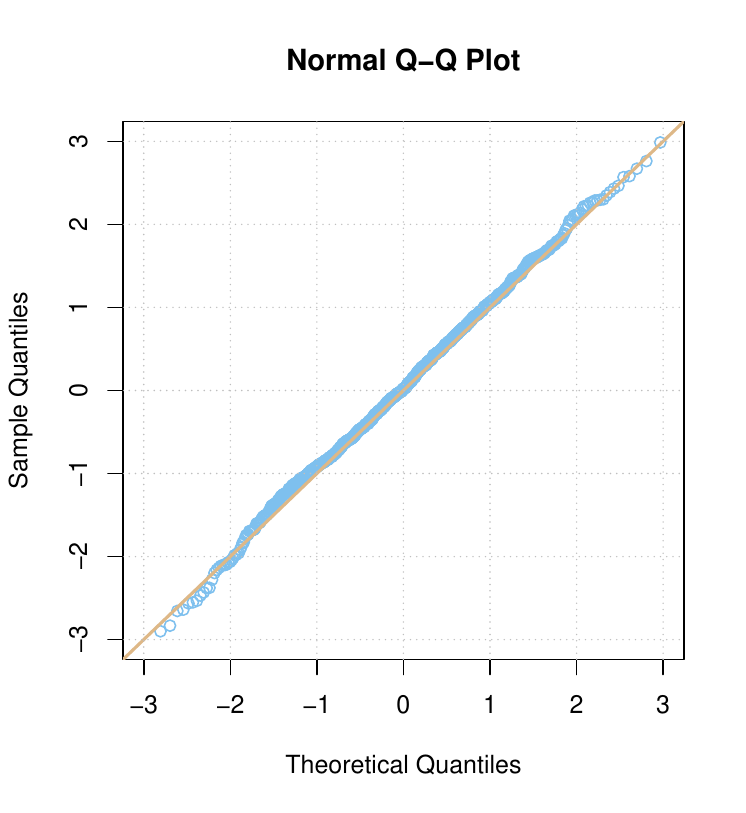}}
	
	\subfigure[$\rho=0.02$]{\label{fig:sbm-sub2.1}
	\includegraphics[width=0.23\linewidth]{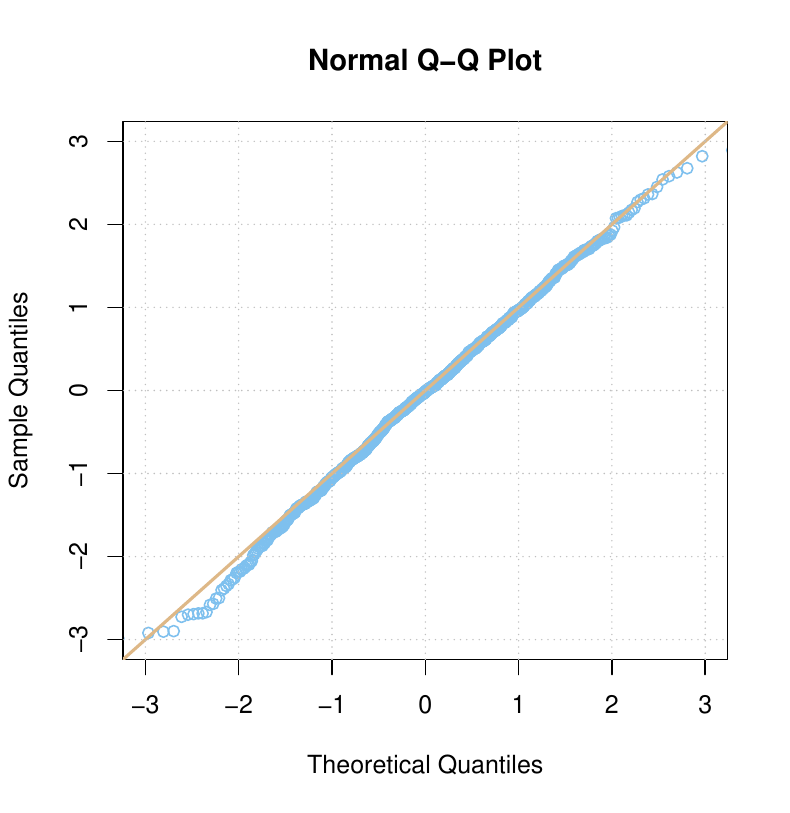}}
	\subfigure[$\rho=0.05$]{\label{fig:sbm-sub2.2}
	\includegraphics[width=0.23\linewidth]{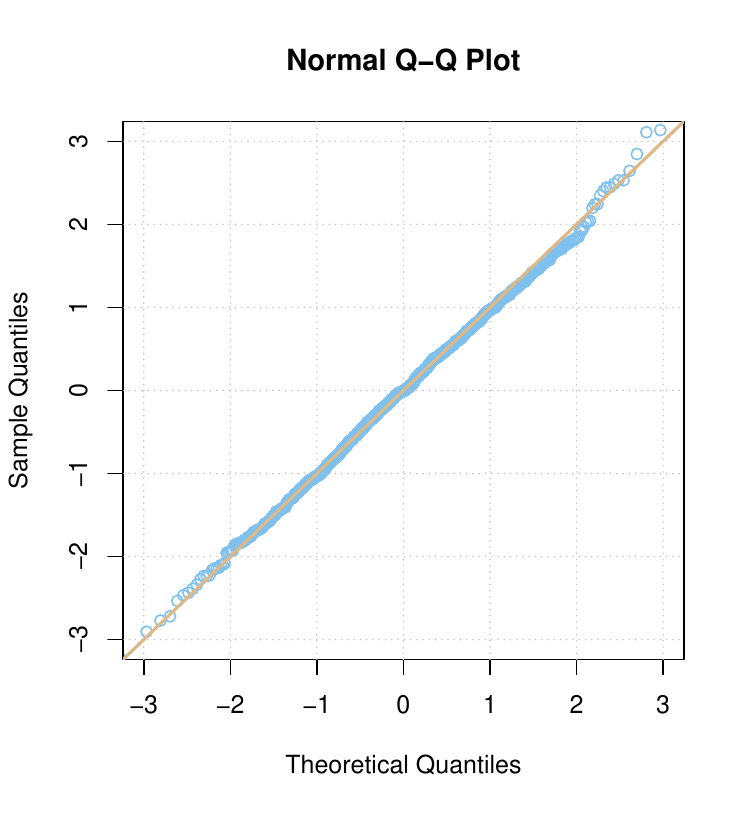}}
	\subfigure[$\rho=0.1$]{\label{fig:sbm-sub2.3}
	\includegraphics[width=0.23\linewidth]{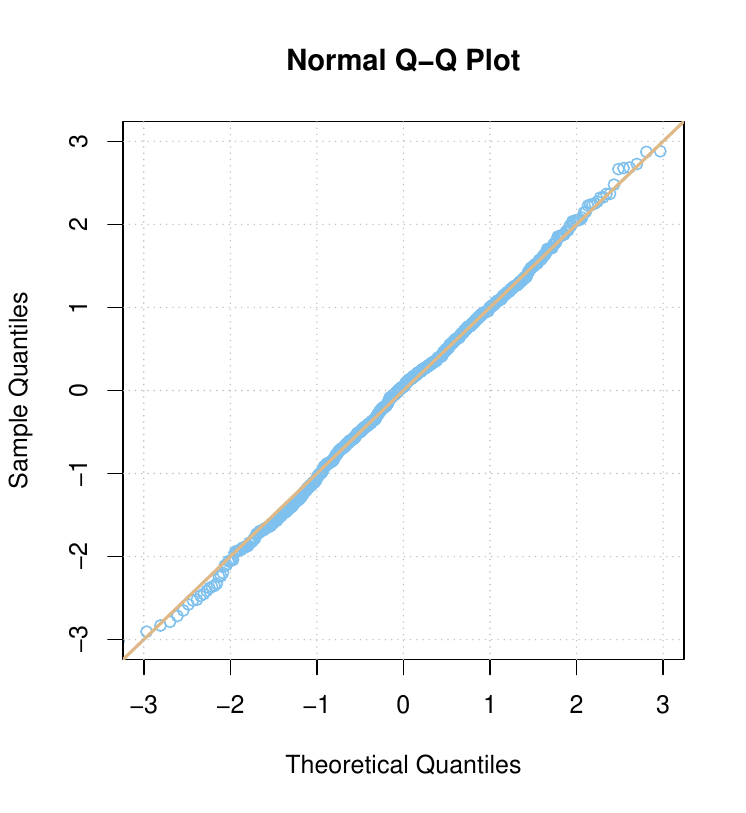}}
	\vspace{0.5cm}
	\caption{Normal Q-Q plot under the SBM when $n=500$ (upper row) and $n=1000$ (lower row).}
	\label{fig:qq-sbm}
\end{figure}

\begin{figure}[htbp]
	\centering
	\vspace{-0.35cm}
	\setlength{\abovecaptionskip}{-2pt}
	\subfigtopskip=2pt
	\subfigbottomskip=2pt
	\subfigcapskip=-5pt
	\subfigure[$\rho=0.02$]{\label{fig:dcsbm-sub1.1}
	\includegraphics[width=0.23\linewidth]{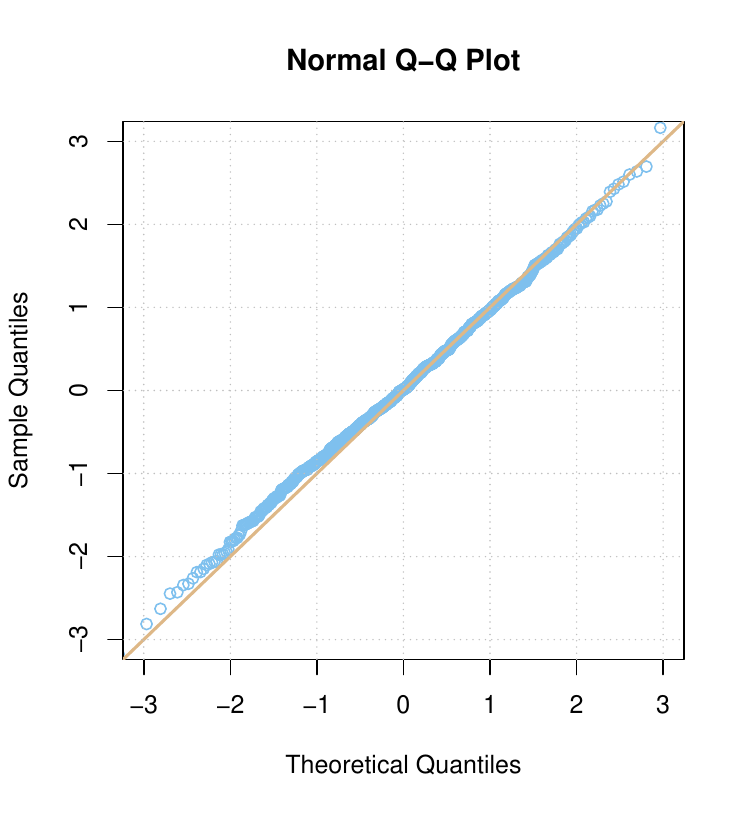}}
	\subfigure[$\rho=0.05$]{\label{fig:dcsbm-sub1.2}
	\includegraphics[width=0.23\linewidth]{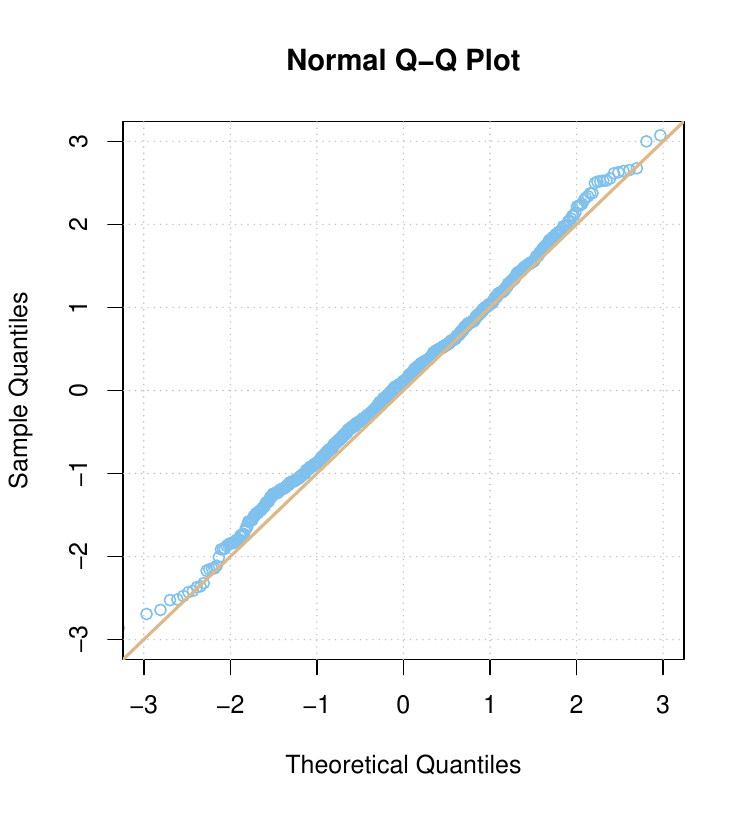}}
	\subfigure[$\rho=0.1$]{\label{fig:dcsbm-sub1.3}
	\includegraphics[width=0.23\linewidth]{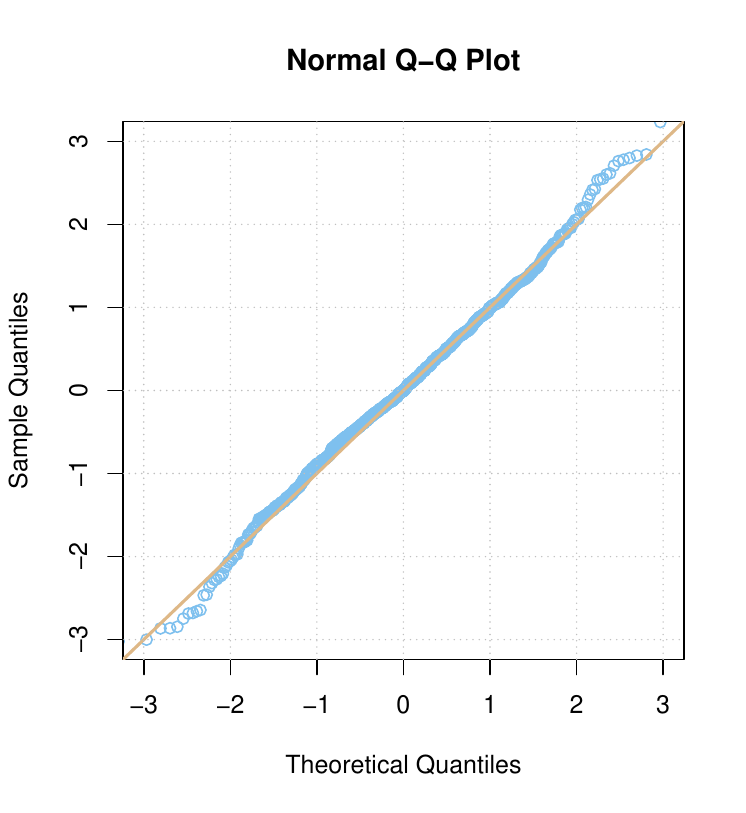}}
	
	\subfigure[$\rho=0.02$]{\label{fig:dcsbm-sub2.1}
	\includegraphics[width=0.23\linewidth]{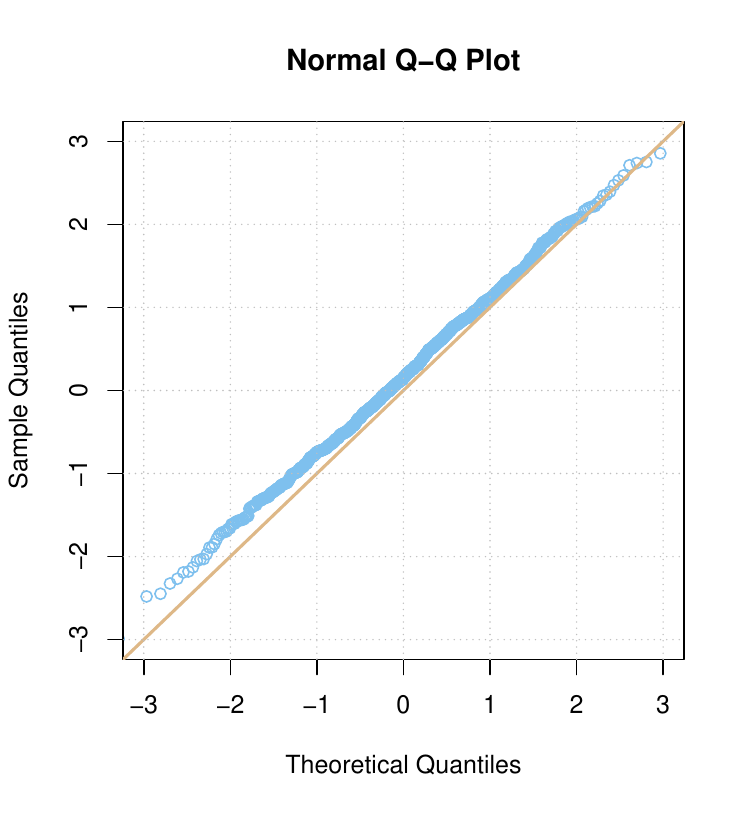}}
	\subfigure[$\rho=0.05$]{\label{fig:dcsbm-sub2.2}
	\includegraphics[width=0.23\linewidth]{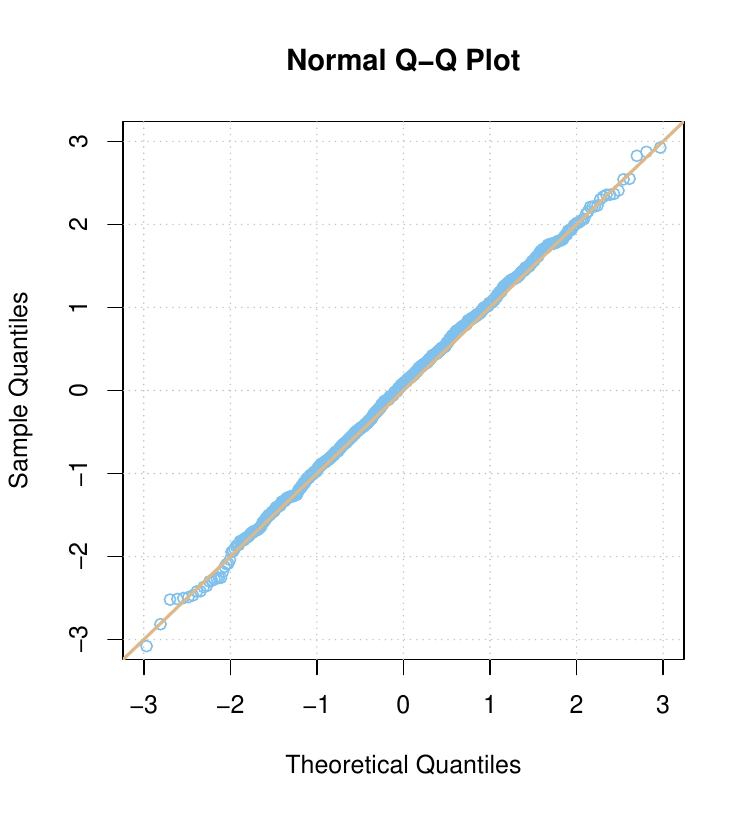}}
	\subfigure[$\rho=0.1$]{\label{fig:dcsbm-sub2.3}
	\includegraphics[width=0.23\linewidth]{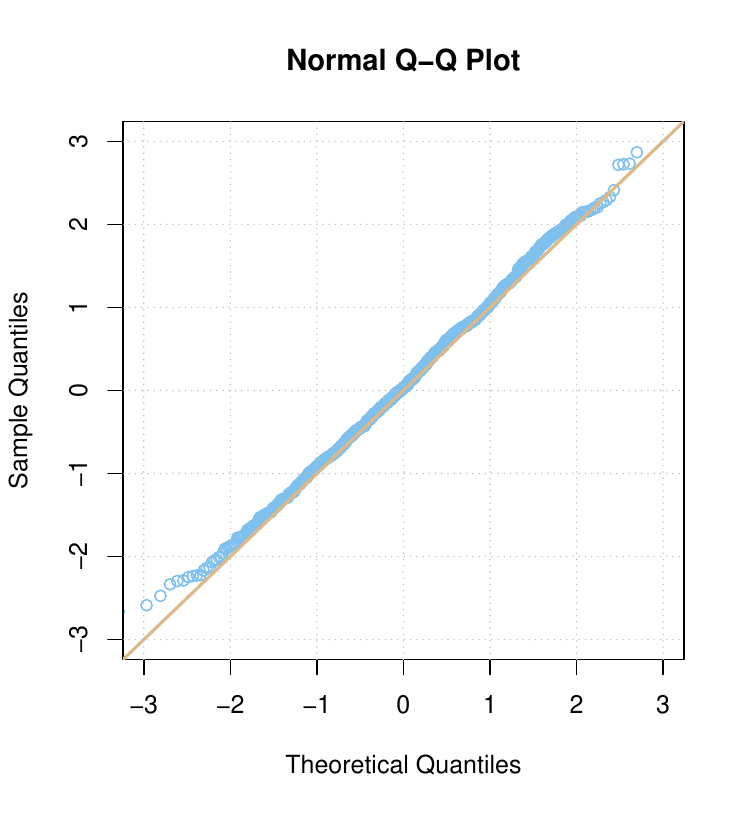}}
	\vspace{0.5cm}
	\caption{Normal Q-Q plot under the DCSBM when $n=500$ (upper row) and $n=1000$ (lower row).}
	\label{fig:qq-dcsbm}
\end{figure}

\begin{figure}[htbp]
	\centering
	\vspace{-0.35cm}
	\setlength{\abovecaptionskip}{-2pt}
	\subfigtopskip=2pt
	\subfigbottomskip=2pt
	\subfigcapskip=-5pt
	\subfigure[$\rho=0.2$]{\label{fig:lsm-sub1.1}
	\includegraphics[width=0.23\linewidth]{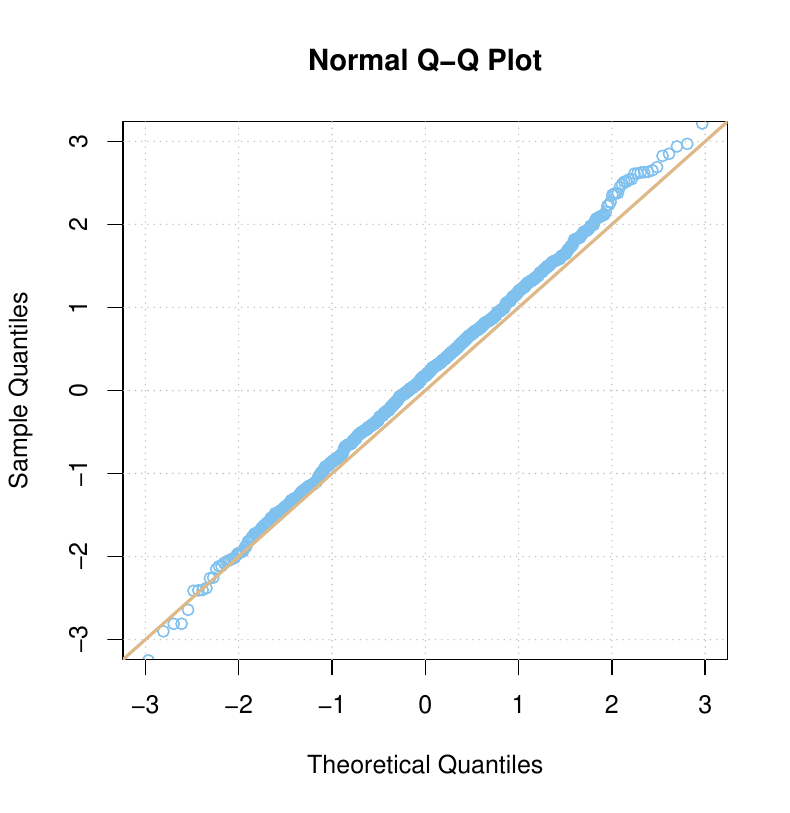}}
	\subfigure[$\rho=0.5$]{\label{fig:lsm-sub1.2}
	\includegraphics[width=0.23\linewidth]{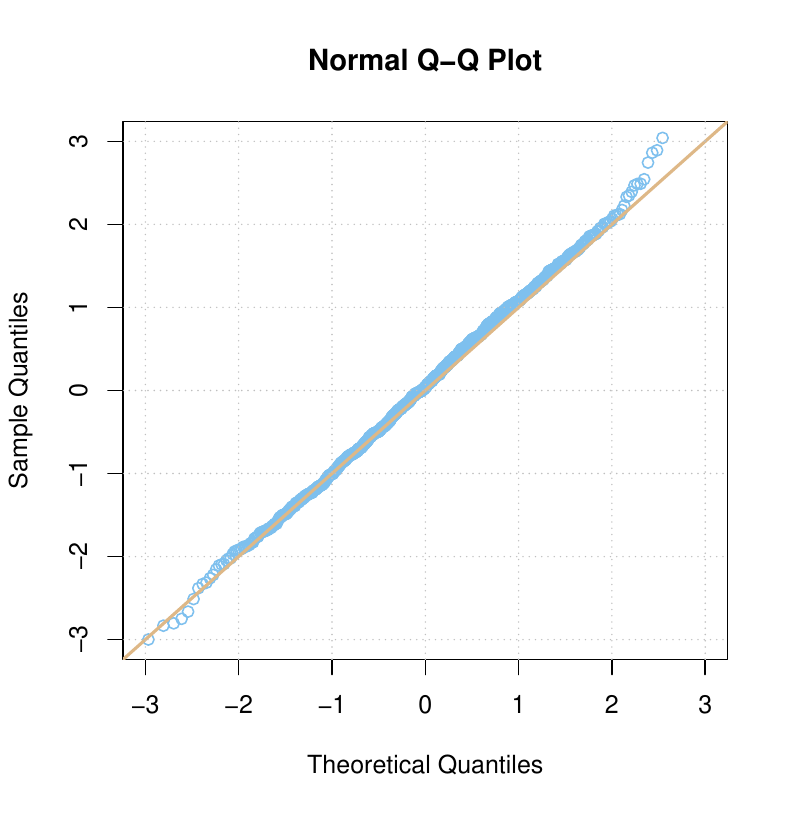}}
	\subfigure[$\rho=1$]{\label{fig:lsm-sub1.3}
	\includegraphics[width=0.23\linewidth]{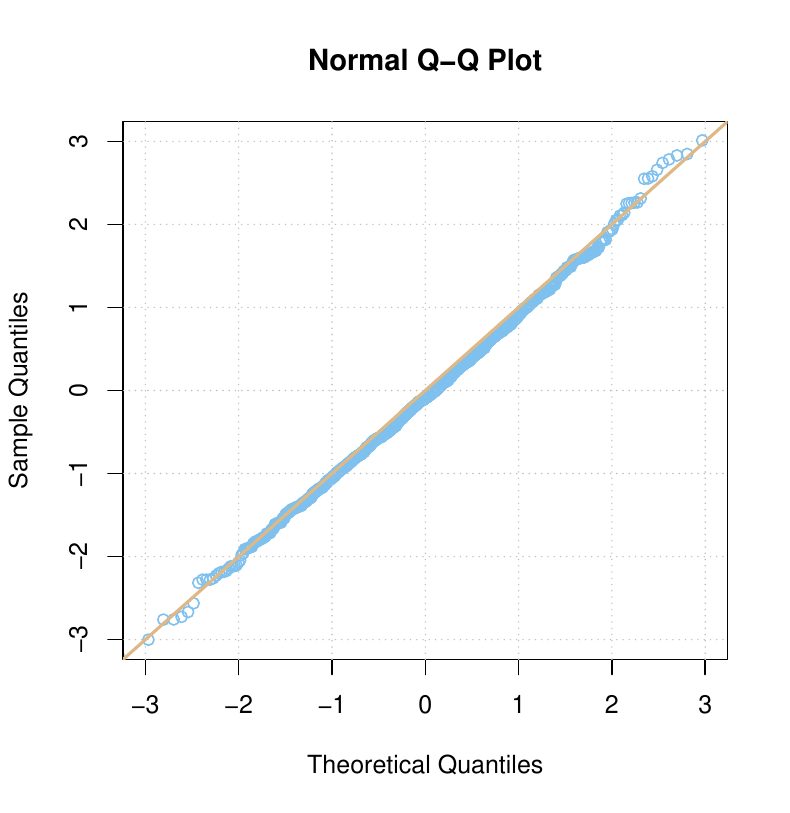}}
	
	\subfigure[$\rho=0.2$]{\label{fig:lsm-sub2.1}
	\includegraphics[width=0.23\linewidth]{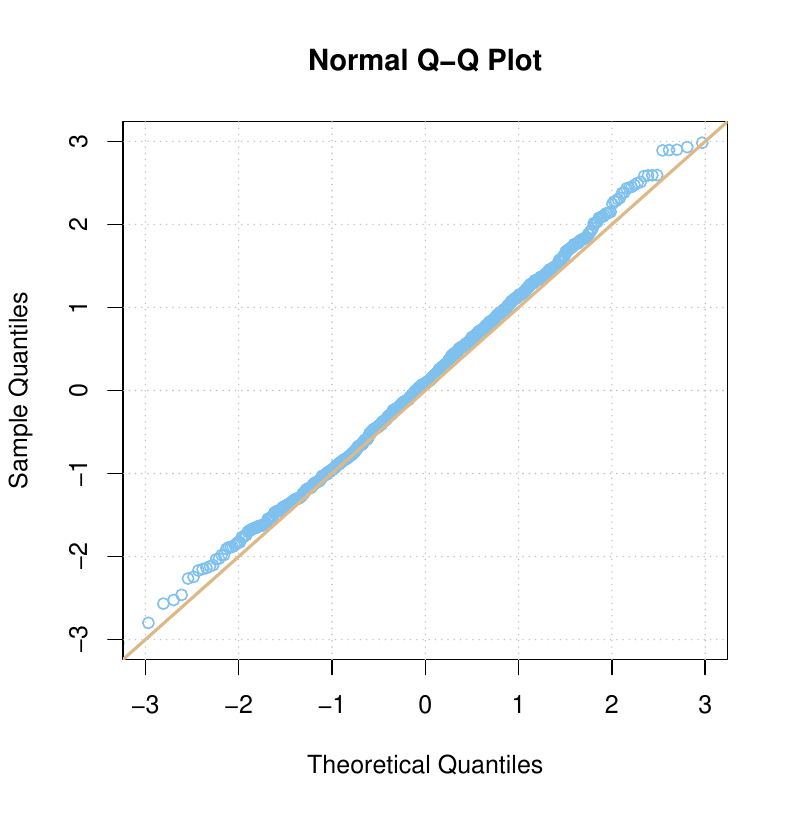}}
	\subfigure[$\rho=0.5$]{\label{fig:lsm-sub2.2}
	\includegraphics[width=0.23\linewidth]{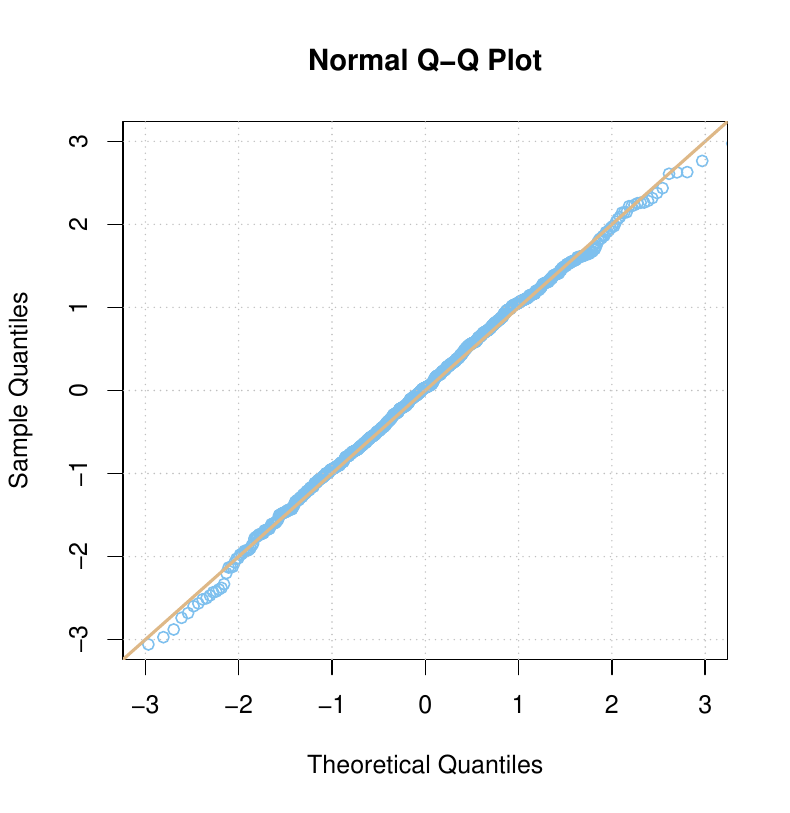}}
	\subfigure[$\rho=1$]{\label{fig:lsm-sub2.3}
	\includegraphics[width=0.23\linewidth]{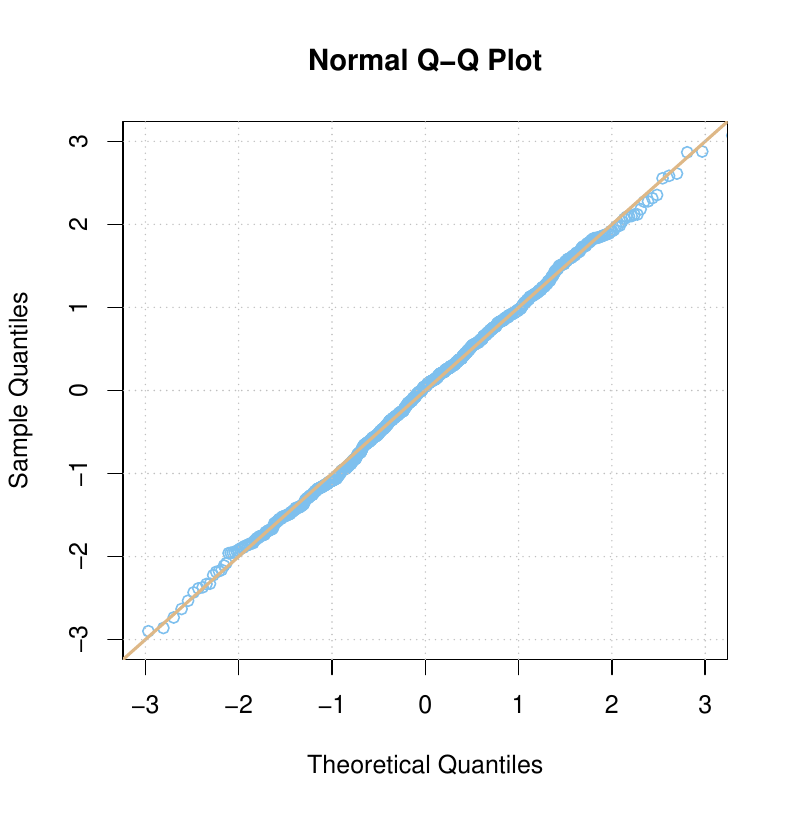}}
	\vspace{0.5cm}
	\caption{Normal Q-Q plot under the LSM when $n=500$ (upper row) and $n=1000$ (lower row).}
	\label{fig:qq-lsm}
\end{figure}

\subsection{The empirical size}

In this subsection, we consider the empirical size. The models and parameter settings are similar to that in Section \ref{sec:null}. Tables \ref{tab:er-size} - \ref{tab:lsm-size} report the results from 200 data replications. From Tables \ref{tab:er-size} - \ref{tab:lsm-size}, for all settings, $T_n$'s Type I errors are close to the nominal level $0.05$. At the same time, it is worth noting that as the sample size increases, the empirical size of the statistic is gradually becoming accurate. The results are consistent with the results in Section \ref{sec:null}.

\begin{table}[htbp]
\setlength{\abovecaptionskip}{0cm}  
\setlength{\belowcaptionskip}{0.5cm} 
\centering
\caption{Empirical size at nominal level $\alpha = 0.05$ for hypothesis test $H_0$ under the E-R model}
\label{tab:er-size}
\begin{tabular*}{\textwidth}{c@{\extracolsep{\fill}}ccc}
\toprule
 & $\rho=0.01$ & $\rho=0.05$ & $\rho=0.1$ \\ \midrule
$n=200$ & 0.01 & 0.03 & 0.07 \\
$n=400$ & 0.06 & 0.05 & 0.04 \\
$n=600$ & 0.03 & 0.05 & 0.05 \\
$n=800$ & 0.06 & 0.06 & 0.05 \\
$n=1000$ & 0.07 & 0.07 & 0.03 \\ \bottomrule
\end{tabular*}
\end{table}

\begin{table}[htbp]
\setlength{\abovecaptionskip}{0cm}  
\setlength{\belowcaptionskip}{0.5cm} 
\centering
\caption{Empirical size at nominal level $\alpha = 0.05$ for hypothesis test $H_0$ under the $\beta$-model}
\label{tab:beta-size}
\begin{tabular*}{\textwidth}{c@{\extracolsep{\fill}}ccc}
\toprule
 & $L_n=0$ & $L_n=(\log(\log n))^{1/3}$ &  $L_n=(\log n)^{1/2}$ \\ \midrule
$n=200$ & 0.03 & 0.06 & 0.04 \\
$n=400$ & 0.03 & 0.05 & 0.05 \\
$n=600$ & 0.06 & 0.06 & 0.04 \\
$n=800$ & 0.07 & 0.05 & 0.03 \\
$n=1000$ & 0.05 & 0.05 & 0.07 \\ \bottomrule
\end{tabular*}
\end{table}

\begin{table}[htbp]
\setlength{\abovecaptionskip}{0cm}  
\setlength{\belowcaptionskip}{0.5cm} 
\centering
\caption{Empirical size at nominal level $\alpha = 0.05$ for hypothesis test $H_0$ under the SBM}
\label{tab:sbm-size}
\begin{tabular*}{\textwidth}{c@{\extracolsep{\fill}}ccc}
\toprule
 & $\rho=0.02$ & $\rho=0.05$ & $\rho=0.1$ \\ \midrule
$n=200$ & 0.68 & 0.05 & 0.03 \\
$n=400$ & 0.06 & 0.03 & 0.04 \\
$n=600$ & 0.05 & 0.04 & 0.06 \\
$n=800$ & 0.05 & 0.06 & 0.06 \\
$n=1000$ & 0.07 & 0.04 & 0.09 \\ \bottomrule
\end{tabular*}
\end{table}

\begin{table}[htbp]
\setlength{\abovecaptionskip}{0cm}  
\setlength{\belowcaptionskip}{0.5cm} 
\centering
\caption{Empirical size at nominal level $\alpha = 0.05$ for hypothesis test $H_0$ under the DCSBM}
\label{tab:dcsbm-size}
\begin{tabular*}{\textwidth}{c@{\extracolsep{\fill}}ccc}
\toprule
 & $\rho=0.02$ & $\rho=0.05$ & $\rho=0.1$ \\ \midrule
$n=200$ & 0.56 & 0.05 & 0.05 \\
$n=400$ & 0.06 & 0.02 & 0.06 \\
$n=600$ & 0.04 & 0.03 & 0.06 \\
$n=800$ & 0.05 & 0.06 & 0.04 \\
$n=1000$ & 0.06 & 0.05 & 0.03 \\ \bottomrule
\end{tabular*}
\end{table}

\begin{table}[htbp]
\setlength{\abovecaptionskip}{0cm}  
\setlength{\belowcaptionskip}{0.5cm} 
\centering
\caption{Empirical size at nominal level $\alpha = 0.05$ for hypothesis test $H_0$ under the LSM}
\label{tab:lsm-size}
\begin{tabular*}{\textwidth}{c@{\extracolsep{\fill}}ccc}
\toprule
 & $\rho=0.2$ & $\rho=0.5$ & $\rho=1$ \\ \midrule
$n=200$ & 0.06 & 0.05 & 0.06 \\
$n=400$ & 0.06 & 0.02 & 0.07 \\
$n=600$ & 0.05 & 0.05 & 0.05 \\
$n=800$ & 0.04 & 0.07 & 0.07 \\
$n=1000$ & 0.03 & 0.02 & 0.04 \\ \bottomrule
\end{tabular*}
\end{table}

\subsection{The empirical power}

In this section, we investigate the empirical power of the proposed test procedure. We consider the following cases:

(i) The true network $A$ is generated from the $\beta$ model with $\beta_i=iL_n/n$ for $1\leq i\leq n$. However, the candidate models $M_1$ are chosen as the E-R model, SBM, and DCSBM.

(ii) The true network $A$ is generated from the SBM with a balanced community and $B_{uv}=\rho(1+4\times I[u=v])$. The candidate models $M_1$ are chosen as the E-R model, $\beta$-model, and LSM with $d=1$.

(iii) The true network $A$ is generated from the DCSBM with a balanced community and $B_{uv}=\rho(1+4\times I[u=v])$. The degree parameters are generated by the method in Section \ref{sec:null}. The candidate models $M_1$ are chosen as the E-R model and LSM with $d=1$ and $d=2$.

(iv) The true network $A$ is generated from the LSM with $d=1$. The candidate models $M_1$ are chosen as the E-R model, SBM, LSM with $d=2$.

\begin{table}[htbp]
\begin{threeparttable}
\setlength{\abovecaptionskip}{0cm}  
\setlength{\belowcaptionskip}{0.5cm} 
\centering
\caption{Empirical power at nominal level $\alpha = 0.05$ for the setting (i).}
\label{tab:power1}
\begin{tabular*}{\textwidth}{c@{\extracolsep{\fill}}cccccccccccc}
\toprule
&& \multicolumn{3}{c}{E-R model} & & \multicolumn{3}{c}{SBM} &  & \multicolumn{3}{c}{DCSBM} \\ \cline{3-5} \cline{7-9} \cline{11-13} 
& $L_n$ & I & II & III &  & I & II & III &  & I & II & III \\ \midrule
$n=200$ && 0.59 & 0.97 & 1 &  & 0.53 & 0.67 & 0.41 &  & 0.63 & 0.96 & 0.92 \\
$n=400$ && 1 & 1 & 1 &  & 0.86 & 0.59 & 0.84 &  & 0.78 & 1 & 1 \\
$n=600$ && 1 & 1 & 1 &  & 0.88 & 0.77 & 0.99 &  & 0.87 & 0.99 & 0.99 \\
$n=800$ && 1 & 1 & 1 &  & 0.90 & 0.98 & 1 &  & 0.84 & 1 & 1 \\
$n=1000$ && 1 & 1 & 1 &  & 0.94 & 1 & 1 &  & 0.63 & 1 & 1 \\ \bottomrule
\end{tabular*}
\begin{tablenotes}[para]
   $^*$ On the second line of Table, I, II, and III indicate $L_n=0$, $L_n=(\log(\log n))^{1/3}$, and  $L_n=(\log n)^{1/2}$, respectively.
\end{tablenotes}
\end{threeparttable}
\end{table}

\begin{table}[htbp]
\setlength{\abovecaptionskip}{0cm}  
\setlength{\belowcaptionskip}{0.5cm} 
\centering
\caption{Empirical power at nominal level $\alpha = 0.05$ for the setting (ii).}
\label{tab:power2}
\begin{tabular*}{\textwidth}{c@{\extracolsep{\fill}}cccccccccccc}
\toprule
&& \multicolumn{3}{c}{E-R model} &  & \multicolumn{3}{c}{$\beta$-model} &  & \multicolumn{3}{c}{LSM$(d=1)$} \\ \cline{3-5} \cline{7-9} \cline{11-13} 
&$\rho$  & $0.02$ & $0.05$ & $0.1$ &  & $0.03$ & $0.05$ & $0.1$ &  & $0.02$ & $0.05$ & $0.1$ \\ \midrule
$n=200$ && 0.64 & 1 & 1 &  & 0.96 & 1 & 1 &  & 0.60 & 1 & 1 \\
$n=400$ && 1 & 1 & 1 &  & 1 & 1 & 1 &  & 1 & 1 & 1 \\
$n=600$ && 1 & 1 & 1 &  & 1 & 1 & 1 &  & 1 & 1 & 1 \\
$n=800$ && 1 & 1 & 1 &  & 1 & 1 & 1 &  & 1 & 1 & 1 \\
$n=1000$ && 1 & 1 & 1 &  & 1 & 1 & 1 &  & 1 & 1 & 1 \\ \bottomrule
\end{tabular*}
\end{table}

\begin{table}[htbp]
\setlength{\abovecaptionskip}{0cm}  
\setlength{\belowcaptionskip}{0.5cm} 
\centering
\caption{Empirical power at nominal level $\alpha = 0.05$ for the setting (iii).}
\label{tab:power3}
\begin{tabular*}{\textwidth}{c@{\extracolsep{\fill}}cccccccccccc}
\toprule
&& \multicolumn{3}{c}{E-R model} &  & \multicolumn{3}{c}{LSM$(d=1)$} &  & \multicolumn{3}{c}{LSM$(d=2)$} \\ \cline{3-5} \cline{7-9} \cline{11-13} 
&$\rho$ & $0.02$ & $0.05$ & $0.1$ &  & $0.02$ & $0.05$ & $0.1$ &  & $0.02$ & $0.05$ & $0.1$ \\ \midrule
$n=200$ && 0.68 & 1 & 1 &  & 0.66 & 1 & 1 &  & 0.34 & 0.87 & 0.98 \\
$n=400$ && 1 & 1 & 1 &  & 1 & 1 & 1 &  & 0.71 & 0.98 & 1 \\
$n=600$ && 1 & 1 & 1 &  & 1 & 1 & 1 &  & 0.90 & 0.99 & 1 \\
$n=800$ && 1 & 1 & 1 &  & 1 & 1 & 1 &  & 0.95 & 1 & 1 \\
$n=1000$ && 1 & 1 & 1 &  & 1 & 1 & 1 &  & 0.96 & 0.99 & 1 \\ \bottomrule
\end{tabular*}
\end{table}

\begin{table}[h]
\setlength{\abovecaptionskip}{0cm}  
\setlength{\belowcaptionskip}{0.5cm} 
\centering
\caption{Empirical power at nominal level $\alpha = 0.05$ for the setting (iv).}
\label{tab:power4}
\begin{tabular*}{\textwidth}{c@{\extracolsep{\fill}}cccccccccccc}
\toprule
&& \multicolumn{3}{c}{E-R model} &  & \multicolumn{3}{c}{SBM} &  & \multicolumn{3}{c}{LSM($d=2$)} \\ \cline{3-5} \cline{7-9} \cline{11-13} 
&$\rho$ & $0.2$ & $0.5$ & $1$ &  & $0.2$ & $0.5$ & $1$ &  & $0.2$ & $0.5$ & $1$ \\ \midrule
$n=200$ && 0.6 & 1 & 1 &  & 0.59 & 1 & 0.99 &  & 0.93 & 0.96 & 1 \\
$n=400$ && 1 & 1 & 1 &  & 1 & 1 & 1 &  & 0.97 & 0.97 & 1 \\
$n=600$ && 1 & 1 & 1 &  & 1 & 1 & 1 &  & 0.98 & 1 & 1 \\
$n=800$ && 1 & 1 & 1 &  & 1 & 1 & 1 &  & 0.98 & 0.99 & 1 \\
$n=1000$ && 1 & 1 & 1 &  & 1 & 1 & 1 &  & 0.99 & 0.99 & 1 \\ \bottomrule
\end{tabular*}
\end{table}

According to the model introduction in Section \ref{sec:method}, there is an inclusion relation between the models, such as the SBM is a special case of the DCSBM. Hence, when considering candidate models, we exclude models that contain true models. Tables \ref{tab:power1} - \ref{tab:power4} report empirical power for different cases from 200 data replications. From Tables \ref{tab:power1} - \ref{tab:power4}, we can observe that the statistic rejects the null hypothesis under all cases. At the same time, it is not hard to see that, with the sample size increasing, the empirical powers are more and more close to 1. 

\subsection{Estimating $K$ for DCMM models}
In the fourth simulation, we examine the performance of the sequential testing estimator of $K$ given in \eqref{eq:EstK} for DCMM models. To estimate the model parameters, we apply the mixed-SCORE algorithm in \cite{Jin:2023}. We restrict the candidate values for the true number of communities in the range $\{1,\cdots,10\}$. We set $K=3$, and $n=500$ or $100$. Given $\rho\in(0,1)$, let the probability matrix $B=\rho\cdot\bm{1}_K\bm{1}_K^\top+(1-\rho)\cdot I_K$. For $0\leq n_0\leq 160$, let each community have $n_0$ number of pure nodes i.e., $\pi_i$ has only one nonzero entry which is equal to 1, and the other entries are zero. For $x\in(0,1/2)$, the rest of the nodes have four different membership vectors $(x,x,1-2x),(x,1-2x,x),(1-2x,x,x)$ and $(1/3,1/3,1/3)$ with equal probability. For the degree parameters, let $1/\theta_i$'s are i.i.d. uniformly random variables in $[1,z]$ for $z\geq 1$. Given the threshold $\alpha=0.001$, Table \ref{tab:EstK} reports the estimation results under the different settings. It is easy to see that the proposed sequential testing method can estimate the number of communities with high accuracy.  According to the results in \cite{Jin:2023}, the higher the fraction of the pure nodes $n_0$, the more accurate the parameter estimation is when other model parameters are fixed. For the different $x$, as $x$ increases to $1/3$, these nodes become less pure; then, as $x$ further increases, these nodes become more pure, which causes the estimation accuracy to decrease first and then increase. For the degree heterogeneous, the larger $z$,  the lower average degree, and the more severe degree heterogeneity. Hence, with the $z$ increasing, the estimation accuracy also corresponding decreases. In addition, as the sample increases, the proportion of correct estimation also increases. Therefore, the numerical results show that the proposed method is an effective and efficient method.

\begin{table}[htbp]
\begin{threeparttable}
\setlength{\abovecaptionskip}{0cm}  
\setlength{\belowcaptionskip}{0.5cm} 
\centering
\caption{Performances of the proposed sequential testing method for estimating the number of\\  communities over 100 simulations.}
\label{tab:EstK}
\begin{tabular}{cc|ccc|ccc}
\toprule
& & $\hat{\bP}\{\hat{K}=K\}$ & $\hat{\bE}\{\hat{K}\}$ & $\widehat{\var}\{\hat{K}\}$ & $\hat{\bP}\{\hat{K}=K\}$ & $\hat{\bE}\{\hat{K}\}$ & $\widehat{\var}\{\hat{K}\}$ \\ \cline{3-8} 
& & \multicolumn{3}{c|}{$n=500$}                                                    & \multicolumn{3}{c}{$n=1000$} \\ \midrule
\multicolumn{1}{c|}{\multirow{4}{*}{\tabincell{c}{$(x,\rho,z)=$\\$(0.4,0.1,5)$}}} & $n_0=40$ & 0.21 & 2.21 & 0.17 & 0.55 & 0.55 & 0.25 \\
\multicolumn{1}{c|}{}                                             & $n_0=80$ & 0.69 & 2.75 & 0.25 & 0.98 & 2.98 & 0.02 \\
\multicolumn{1}{c|}{}                                             & $n_0=120$ & 0.85 & 2.99 & 0.15 & 1.00 & 3.00 & 0.00 \\
\multicolumn{1}{c|}{}                                             & $n_0=160$ & 0.98 & 2.98 & 0.02 & 1.00 & 3.00 & 0.00 \\ \midrule
& & \multicolumn{3}{c|}{$n=500$}                                                    & \multicolumn{3}{c}{$n=1000$} \\ \midrule
\multicolumn{1}{c|}{\multirow{4}{*}{\tabincell{c}{$(n_0,\rho,z)=$ \\$(80,0.1,5)$}}} & $x=0.05$ & 0.83 & 2.99 & 0.17 & 1.00 & 3.00 & 0.00 \\
\multicolumn{1}{c|}{}                                             & $x=0.15$ & 0.71 & 2.77 & 0.24 & 0.99 & 2.99 & 0.01 \\
\multicolumn{1}{c|}{}                                             & $x=0.25$ & 0.45 & 2.51 & 0.31 & 0.98 & 3.00 & 0.02  \\
\multicolumn{1}{c|}{}                                             & $x=0.35$ & 0.73 & 2.81 & 0.23 & 0.96 & 3.04 & 0.04  \\ \midrule
& & \multicolumn{3}{c|}{$n=500$}                                                    & \multicolumn{3}{c}{$n=1000$}                                                    \\ \midrule
\multicolumn{1}{c|}{\multirow{4}{*}{\tabincell{c}{$(x,n_0,\rho)=$\\$(0.4,80,0.1)$}}} & $z=1$ & 1.00 & 3.00 & 0.00 & 1.00 & 3.00 & 0.00 \\
\multicolumn{1}{c|}{}                                             & $z=3$ & 0.98 & 2.98 & 0.02 & 1.00 & 3.00 & 0.00 \\
\multicolumn{1}{c|}{}                                             & $z=5$ & 0.64 & 2.66 & 0.25 & 0.94 & 2.96 & 0.06 \\
\multicolumn{1}{c|}{}                                             & $z=7$ & 0.48 & 2.58 & 0.35 & 0.74 & 2.84 & 0.24 \\ 
\bottomrule
\end{tabular}
\begin{tablenotes}[para]
   $^*$ Since the fraction of pure node depends on the sample size $n$, the $n_0$ when $n=1000$ is twice as much as when $n = 500$ under other settings are the same, such as set $n_0=40$ and $n_0=80$ when $n=500$ and $n=1000$, respectively, in the setting of the third line of the table.
\end{tablenotes}
\end{threeparttable}
\end{table}

\section{Real analysis}\label{sec:real}

In this section, we apply the proposed method to five real network datasets. The first dataset is the food web dataset. This dataset is from \cite{Baird:1989} and is available in \cite{Blitzstein:2011}, which contains data on 33 organisms (such as bacteria, oysters, and catfish) in the Chesapeake Bay during the summer. The karate club network of \cite{Zachary:1977} is a social network of 34 members of a karate club at a US university, with edges representing friendship patterns. The dolphin network collected by \cite{Lusseau:2003} is an undirected social network of frequent associations between 62 dolphins in a community living off Doubtful Sound, New Zealand. The college football network is derived from the schedule of Division I games for the 2000 season in the United States \citep{Girvan:2002}. It has 115 nodes, representing the football teams, and 613 edges, indicating regular-season games between pairs of teams. The international trade dataset originally analyzed in \cite{Westveld:2011} contains yearly international trade data between 58 countries from 1981 to 2000. For this network, an adjacency matrix $A$ can be formed by first considering a weight matrix $W$ with $W_{ij} = Trade_{i,j} + Trade_{j,i}$, where $Trade_{i,j}$ denotes the value of exports from country $i$ to country $j$. Finally, we define $A_{ij} = 1$ if $W_{ij} \geq W_{0.5}$, and $A_{ij} = 0$ otherwise; here $W_{0.5}$ denotes the $50\%$-th quantile of $\{W_{ij}\}_{1\leq i<j\leq n}$. Food web dataset and trade dataset are from \cite{Blitzstein:2011} and \cite{Westveld:2011}, and the other 3 datasets are downloaded from \url{http://www-personal.umich.edu/~mejn/netdata/}). Table \ref{tab:realdata} reports the number of nodes, the number of communities, the number of edges, and the node degree for the 5 network datasets. Note that the network can be seen as there is a severe degree heterogeneity when $d_{max}/d_{min}$ are as large as a few hundred. 

Based on the proposed method, Table \ref{tab:realpvalue} reports the $p$-values of the test for the 5 networks. The results show that every network can fit a network model. It is worth noting that the $p$-values of the E-R model and SBM are equal for the foodweb network and karate network. This is because, under the SBM, the number of communities is estimated as 1, and then the model is reduced to the E-R model. Hence, we can consider the foodweb network to be from the E-R model. For the last 4 data, the previous studies have shown that the networks have a community structure, and the proposed test procedure can also accept the SBMs and DCSBMs can fit these data under different nominal levels. 

Finally, we apply the method in Section \ref{sec:appdcmm} to the last four datasets. The sequential estimation $\hat{K}$ are 2, 4, 11, and 2 for the last four networks karate, dolphin, football, and trade, respectively. The sequential estimations $\hat{K}$ are consistent with the true number of communities, except for the trade data is underestimated. Figure \ref{fig:real} plots these network visualizations. All analysis results show that the proposed procedure is an effective and efficient method.

\begin{table}[htbp]
\setlength{\abovecaptionskip}{0cm}  
\setlength{\belowcaptionskip}{0.5cm} 
\centering
\caption{The 5 network data sets we analyze in this paper. ($d_{min}$, $d_{max}$, and $\bar{d}$ stand for the minimum degree, maximum degree, and average degree, respectively).}
\label{tab:realdata}
\begin{tabular*}{\textwidth}{c@{\extracolsep{\fill}}cccccc}
\toprule
Dataset & $n$ & $K$ & $\#$edges & $d_{max}$ & $d_{min}$ & $\bar{d}$ \\ \midrule
foodweb & 33 & - & 71 & 1 & 10 & 4.30 \\ 
karate & 34 & 2 & 78 & 1 & 17 & 4.59 \\
dolphin & 62 & 2,4 & 159 & 1 & 12 & 5.12 \\
football & 110 & 11 & 570 & 7 & 13 & 10.36 \\
trade & 58 & 3 & 841 & 3 & 57 & 29 \\ \bottomrule
\end{tabular*}
\end{table}

\begin{figure}[htbp]
	\centering
	\vspace{-0.35cm}
	\setlength{\abovecaptionskip}{-2pt}
	\subfigtopskip=2pt
	\subfigbottomskip=2pt
	\subfigcapskip=-5pt
	
	\subfigure[The food web network]{\label{fig:foodweb}
	\includegraphics[width=0.4\linewidth]{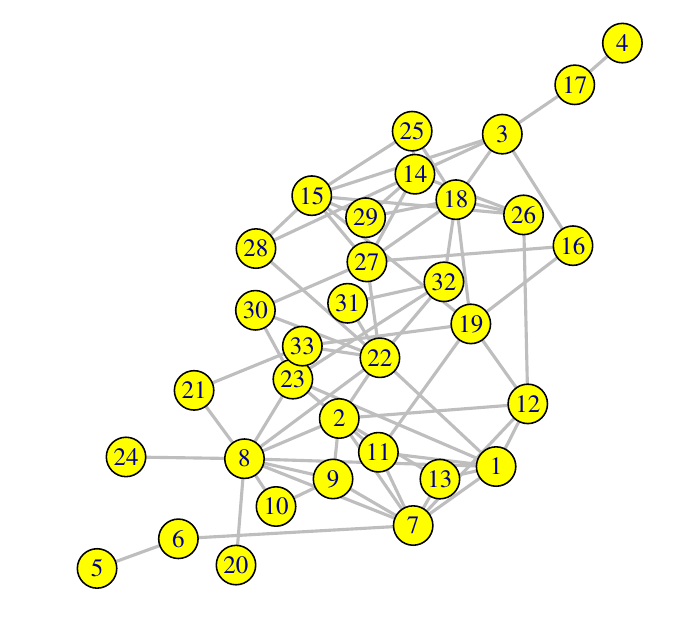}}
	
	\subfigure[The karate club network]{\label{fig:karate}
	\includegraphics[width=0.4\linewidth]{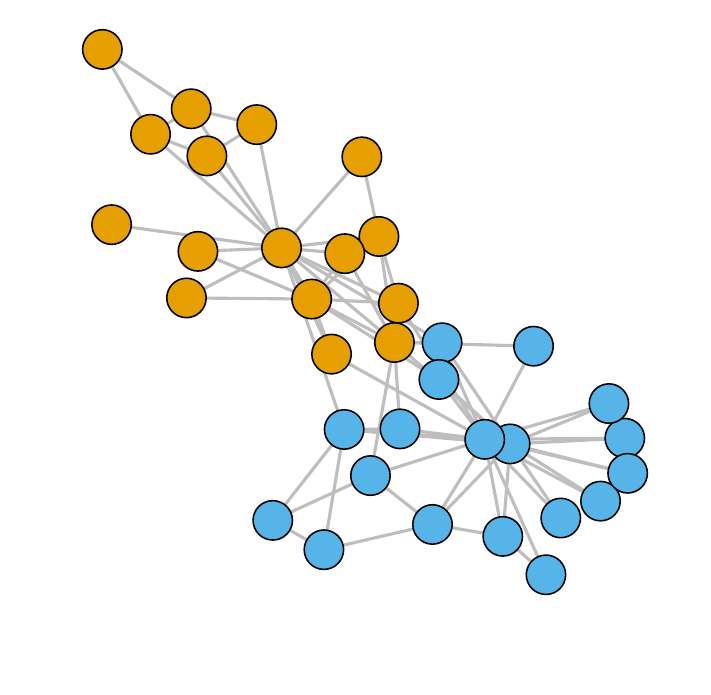}}
	\subfigure[The dolphin social network]{\label{fig:dolphin}
	\includegraphics[width=0.4\linewidth]{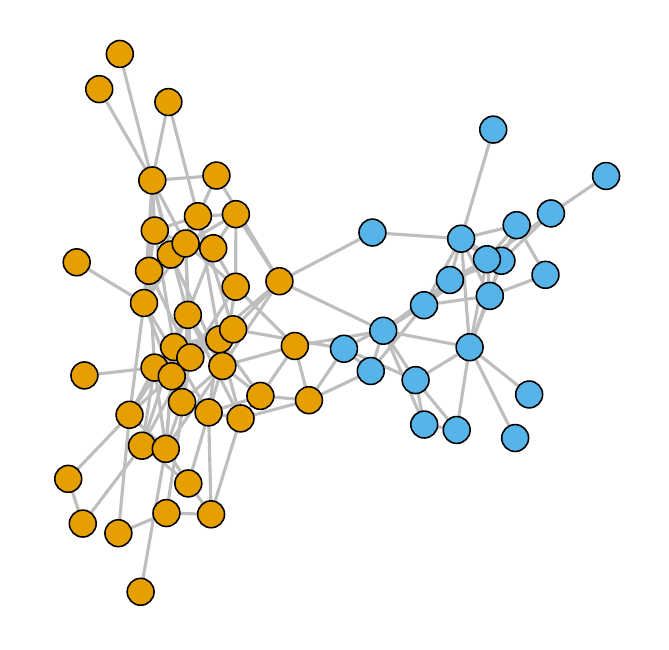}}
	
	\subfigure[The college football network]{\label{fig:football}
	\includegraphics[width=0.4\linewidth]{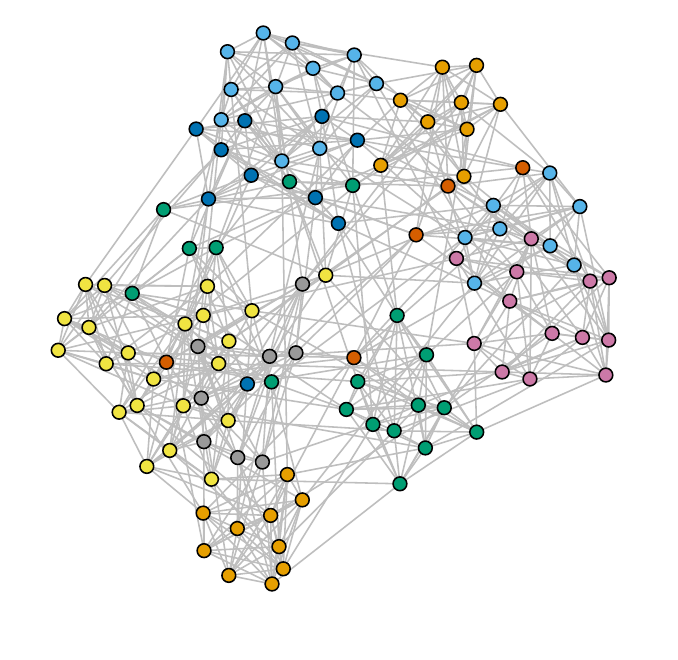}}
	\subfigure[The international trade network]{\label{fig:trade}
	\includegraphics[width=0.4\linewidth]{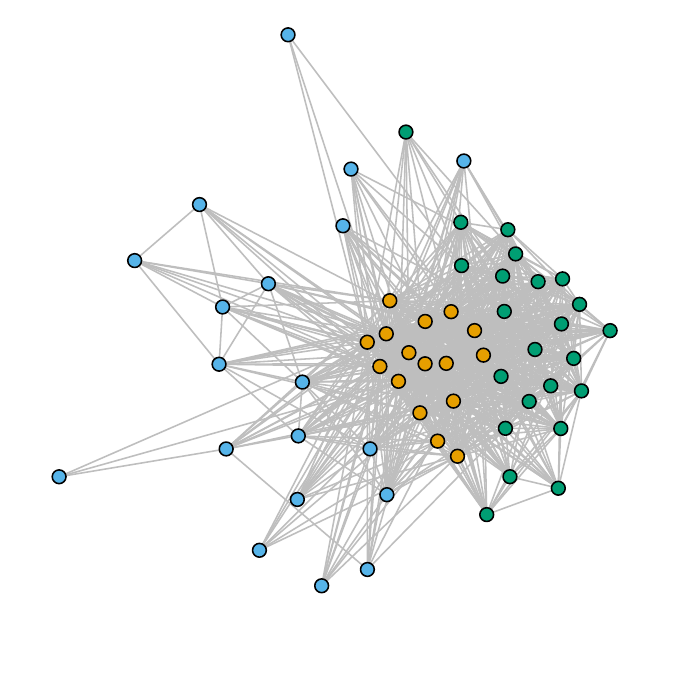}}
	
	\vspace{0.5cm}
	\caption{The network visualizations for five real network datasets. For the last 4 networks, the nodes with common colors are clustered into one group.}
	\label{fig:real}
\end{figure}

\begin{table}[htbp]
\setlength{\abovecaptionskip}{0cm}  
\setlength{\belowcaptionskip}{0.5cm} 
\centering
\caption{The $p$-values for 5 network data sets under the different model.}
\label{tab:realpvalue}
\begin{threeparttable}
\begin{tabular*}{\textwidth}{c@{\extracolsep{\fill}}ccccc}
\toprule
Dataset & E-R model & $\beta$-model & SBM & DCSBM & LSM \\ \midrule
foodweb & $\bm{0.9362}$ & $0.6393$ & $0.9362$ & $0.6699$ & $0.0031$ \\ 
karate & $0.2625$ & $4.4146\times10^{-7}$ & $0.2625$ & $\bm{0.7866}$ & $2.4699\times10^{-15}$ \\
dolphin & $1.5733\times10^{-48}$ & $4.0769\times10^{-14}$ & $1.6186\times10^{-6}$ & $\bm{0.0007}$ & $0$ \\
football & $0$ & $0$ & $\bm{0.0383}$ & $0.0003$ & $0$ \\
trade & $4.7139\times10^{-45}$ & $0$ & $\bm{0.2248}$ & $0.0462$ & $2.8909\times10^{-6}$ \\ \bottomrule
\end{tabular*}
\begin{tablenotes}[para]
   $^*$ The bold values represent the $p$-value corresponding to the most suitable model for the network data.
\end{tablenotes}
\end{threeparttable}
\end{table}

\section{Conclusion}\label{sec:conclusion}

In this paper, we have proposed a novel spectral-based statistic to investigate the goodness-of-fit test for the general network model. Based on the random matrix theory, we have proved the limiting distribution of trace of the third-order of the normalized adjacency matrix is a normal distribution. Further, plugging in an estimate of parameters, we have also proved the limiting distribution of trace of the third-order of the empirically normalized adjacency matrix is also a normal distribution under some mild conditions. Empirically, we have demonstrated that the size and the power of the test are valid.

In some network models, as a technical reason, the errors of the parameter estimation are not fully investigated, such as DCMM models. Hence, the theoretical properties of the proposed statistic are difficult to be obtained under the existing method. In addition, based on the proposed method, we give a sequential estimation for the number of communities $K$ in DCMM models. However, the consistency is not proven. One of the main obstacles is the non-identifiability of the model. Further, we can also consider estimating the dimension $d$ for the LSM, which is an interesting and under-explored issue. We leave the detailed formulation and theoretical study to future work.

\section{Appendix}\label{sec:appendix}

\subsection{Results from random matrix theory}

In this section, we first collect some useful results from random matrix theory (RMT) regarding the Wigner matrix.

Let $W_n$ be a $n\times n$ Wigner matrix with eigenvalues $\lambda_1(W_n),\ldots,\lambda_n(W_n)$. Since $W_n$ is a Hermitian matrix, then all eigenvalues are real. The empirical spectral distribution (ESP) of $W_n$ is as follows:
\[
F^{W_n}(x) = \dfrac{1}{n}\sum_{i=1}^n I[\lambda_i(W_n)\leq x].
\]

A Wigner matrix is a Hermitian random matrix whose elements on or above the diagonal are independent. Suppose $W_n$ is an $n\times n$ Hermitian matrix whose diagonal elements are i.i.d. random variables and those above the diagonal are i.i.d. random variables with mean 0 and variance 1. Then, \cite{Bai:2016} proved that the ESD of normalized Wigner matrix $X_n=W_n/\sqrt{n}$ tends to the semicircular law
\[
F(x) = \dfrac{1}{2\pi}\sqrt{4-x^2},\qquad x\in[-2,2],
\]
with probability 1.

Let $\mathscr{U}$ be an open set of the real line that contains the interval $[-2,2]$. Further, define $\mathscr{F}$ to be the set of analytic functions $f:\mathscr{U}\mapsto\bR$. Then, we mainly consider the empirical process $\mathcal{G}_n=\{\mathcal{G}_n(f)\}$ indexed by $\mathscr{F}$; i.e.,
\[
\mathcal{G}_n(f) = n\int_{\bR}f(x)[F^{X_n}-F](dx),\qquad f\in\mathscr{F}.
\]
To study the limiting behavior of $\mathcal{G}_n$, the following conditions on the moments of the entries $W_{ij}$ of Wigner matrix $W_n$ with $\bE\{W_{ij}\}=0$ are given in \cite{Wang:2021}:
\begin{enumerate}
	\item The random variables $W_{ij}'s$ are uniformly bounded in any $L^p$ space ($p\geq 1$). 
	\item For all $i$, $\bE\{|W_{ii}|^2\}=\sigma^2$, for all $i<j$, $\bE\{|W_{ij}|^2\}=1$;
	\item $\dfrac{1}{n^2}\sum_{i,j}\bE\{|W_{ij}|^4\}\rightarrow M$;
	\item For any $\eta>0$, as $n\rightarrow\infty$,
	      \[
	      \dfrac{1}{\eta^4n^2}\sum_{i,j}\bE\{|W_{ij}|^4 I[|W_{ij}|\geq\eta\sqrt{n}]\} = o(1).
	      \]
\end{enumerate}

Define, for $f\in\mathscr{F}$ and any integer $\ell\geq0$, $\tau_\ell(f)=\dfrac{1}{2\pi}\int_{-\pi}^{\pi}f(2\cos\theta)e^{i\ell\theta}d\theta$. Then \cite{Wang:2021} given the following theorem:
\begin{lemma}[Theorem 2.1 in \cite{Wang:2021}]\label{lemma:wang}
	Under conditions [C1]--[C4], the spectral empirical process $\mathcal{G}_n=\{\mathcal{G}_n(f)\}$ indexed by the set of analytic functions $\mathscr{F}$ converges weakly in finite dimension to a Gaussian process $\mathcal{G}=\{\mathcal{G}(f):f\in\mathscr{F}\}$ with mean function $\bE(\mathcal{G}(f))$ given by
	\[
	\dfrac{1}{4}\left\{f(2)+f(-2)\right\}-\dfrac{1}{2}\tau_0(f)+(\sigma^2-2)\tau_2(f)+(M-3)\tau_4(f)
	\]
	and the covariance function $c(f,g)=\bE\{[\mathcal{G}(f)-\mathcal{G}(g)][\mathcal{G}(f)-\mathcal{G}(g)\}$ given by
	\[
	(\sigma^2-2)\tau_1(f)\tau_1(g)+2(M-3)\tau_2(f)\tau_2(g)+2\sum_{\ell=1}^{\infty}\ell\tau_\ell(f)\tau_\ell(g)=\dfrac{1}{4\pi^2}\int_{-2}^{2}\int_{-2}^{2}f^{\prime}(t)g^{\prime}(s)V(t,s)dtds,
	\]
	where $V(t,s)=\left(\sigma^2-2+\dfrac{1}{2}(M-3) ts\right)\sqrt{(4-t^2)(4-s^2)}+2\log\left(\dfrac{4-ts+\sqrt{(4-t^2)(4-s^2)}}{4-ts-\sqrt{(4-t^2)(4-s^2)}}\right)$.
\end{lemma}

Next, this theorem will be the main tool for the proof of the main result.

\begin{lemma}\label{lemma:1}
	For $\tilde{A}^*$ defined in \eqref{eq:spec1}, we have
	\[
	\dfrac{1}{\sqrt{6}}\tr((\tilde{A}^{*})^3)\rightsquigarrow N(0,1).
	\]
\end{lemma}

\begin{proof}
	First, we verify that $W=\sqrt{n}\tilde{A}^{*}$ satisfies the conditions [C1]--[C4]. Condition [C1] is a trivial fact, due to $W_{ii} = 0$ and $W_{ij} = \dfrac{A_{ij}-p_{ij}}{p_{ij}(1-p_{ij})}$. For Condition [C2], it is not difficult to see that, for all $1\leq i,j\leq n$, $\bE\{|W_{ii}|^2\}=0$ and $\bE\{|W_{ij}|^2\}=n\bE\{|\tilde{A}^{*}_{ij}|^2\}=1$. For Condition [C3], we have $\bE\{|W_{ij}|^4\}=n^2\bE\{|\tilde{A}^{*}_{ij}|^4\}=O(1)$ and $\sum_{ij}\bE\{|W_{ij}|^4\} = O(n^2)$. Next, we verify Condition [C4]. For any $\eta>0$, we have
	\begin{align*}
		&~ \dfrac{1}{\eta^4n^2}\sum_{i,j}\bE\{|W_{ij}|^4 I[|W_{ij}|\geq\eta\sqrt{n}]\}\\
		\leq &~ \dfrac{1}{\eta^4n^2}\sum_{i,j}\left(\bE\{|W_{ij}|^8\}\right)^{1/2}\left(\bP\{|W_{ij}|\geq\eta\sqrt{n}\}\right)^{1/2}\\
		\leq &~ \dfrac{C}{\eta^4}\max_{i,j}\left\{\left(\bP\{|W_{ij}|\geq\eta\sqrt{n}\}\right)^{1/2}\right\}\\
		= &~ \dfrac{C}{\eta^4}\max_{i\neq j}\left\{\left(\bP\left\{\left|\dfrac{A_{ij}-p_{ij}}{\sqrt{p_{ij}(1-p_{ij})}}\right|\geq\eta\sqrt{n}\right\}\right)^{1/2}\right\}\\
		= &~ o(1),
	\end{align*}
	where $C$ is a constant that upper bound $\left(\bE\{|X_{ij}|^8\}\right)^{1/2}$, and the last equality is due to $W_{ij}$'s are uniformly bounded. Set $f(x)=x^3\in\mathscr{F}$, then we have
	\begin{align*}
		\mathcal{G}_n(f) & = n\int_{\bR}f(x)[F_n-F](dx) \\
		& =	n\int_{\bR}x^3F_n(dx)-n\int_{\bR}x^3F(dx) \\
		& = \sum_{i=1}^n[\lambda_i(\tilde{A}^{*})]^3 \\
		& = \tr((\tilde{A}^{*})^3).
	\end{align*}
	Finally, following Lemma \ref{lemma:wang}, it is easy to get $\bE(\mathcal{G}(x^3))=0$ and $\var(\mathcal{G}(x^3))=6$.
\end{proof}

\subsection{Proof of Theorem \ref{thm:null}}
Let $\tilde{A}^{\prime}\in\bR^{n\times n}$ be such that
\begin{equation*}
	\tilde{A}^{\prime}_{ij} = \begin{dcases}\dfrac{A_{ij}-\hat{p}_{ij}}{\sqrt{np_{ij}(1-p_{ij})}}, & i\neq j,\\ 0, & i=j.\end{dcases}
\end{equation*}
Thus, we have $\tilde{A}^{\prime}=\tilde{A}^*+\Delta^{\prime}$.

Then, we consider the difference between $\tr((\tilde{A}^{\prime})^3)$ and $\tr((\tilde{A}^{*})^3)$. It is easy to see that
\begin{equation}\label{eq:app6}
\tr((\tilde{A}^{\prime})^3-(\tilde{A}^{*})^3)=3\tr((\tilde{A}^{*})^2\Delta^{\prime})+3\tr(\tilde{A}^{*}(\Delta)^2)+\tr((\Delta^{\prime})^3).
\end{equation}

For the first term of equality \eqref{eq:app6}, we have
\begin{align}\label{eq:app7}
	\tr((\tilde{A}^{*})^2\Delta^{\prime}) & = \sum_{i,j,k}\tilde{A}^{*}_{ij}\tilde{A}^{*}_{jk}\Delta_{ki}^{\prime} \notag \\
	& = \sum_{j} \sum_{k,i}\tilde{A}^{*}_{jk}\Delta_{ki}^{\prime}\tilde{A}^{*}_{ij} \notag \\
	& = \sum_{k\neq i}\Delta_{ki}^{\prime}\sum_{j}\tilde{A}^{*}_{jk}\tilde{A}^{*}_{ij}
\end{align}

Further, we have
\begin{align*}
\bE\left\{\sum_{k\neq i}\Delta_{ki}^{\prime}\sum_{j}\tilde{A}^{*}_{jk}\tilde{A}^{*}_{ij}\right\} & = 2\bE\left\{\bE\left\{\sum_{k>i}\Delta_{ki}^{\prime}\sum_{j}\tilde{A}^{*}_{jk}\tilde{A}^{*}_{ij} \vert \Delta^{\prime}\right\}\right\}\notag \\
& = 2\bE\left\{\sum_{k>i}\Delta_{ki}^{\prime}\sum_{j}\bE\{\tilde{A}^{*}_{jk}\tilde{A}^{*}_{ij}\}\right\}\notag \\
& = 0,
\end{align*}
and
\begin{align*}
	\var\left\{\sum_{k\neq i}\Delta_{ki}^{\prime}\sum_{j}\tilde{A}^{*}_{jk}\tilde{A}^{*}_{ij}\right\} & = \var\left\{\bE\left\{\sum_{k\neq i}\Delta_{ki}^{\prime}\sum_{j}\tilde{A}^{*}_{jk}\tilde{A}^{*}_{ij}\vert\Delta^{\prime}\right\}\right\} \\
	& \qquad + \bE\left\{\var\left\{\sum_{k\neq i}\Delta_{ki}^{\prime}\sum_{j}\tilde{A}^{*}_{jk}\tilde{A}^{*}_{ij}\vert\Delta^{\prime}\right\}\right\} \\
	& = 4\bE\left\{\sum_{k>i}(\Delta_{ki}^{\prime})^2\sum_{j}\var\{\tilde{A}^{*}_{jk}\tilde{A}^{*}_{ij}\}\right\} \\
	& \leq 4\bE\left\{\sum_{k>i}(\Delta_{ki}^{\prime})^2\sum_{j}\bE\{(\tilde{A}_{kj}^{*})^2\}\bE\{(\tilde{A}_{lj}^{*})^2\}\right\}\\
	& = o_p(n^{-1/2}).
\end{align*}
The last line use result that $\max_{ij}|\Delta_{ij}^{\prime}|=o_p(n^{-3/4})$. Hence, $\tr((\tilde{A}^{*})^2\Delta^{\prime}) = o_p(n^{-1/2})$.

For the second term of equality \eqref{eq:app6}, we have
\[
\tr(\tilde{A}^{*}(\Delta^{\prime})^2)=\tr(\tilde{A}^{*}\Gamma^\top\Lambda^2\Gamma)=\tr(\Lambda^2\Gamma\tilde{A}^{*}\Gamma^\top).
\]
According to the algebra calculation, we have
\begin{align*}
	\Gamma\tilde{A}^{*}\Gamma^\top & = \begin{pmatrix}\Gamma_1^\top \\ \vdots\\
	\Gamma_n^\top\end{pmatrix}\tilde{A}^{*}\begin{pmatrix}\Gamma_1 & \cdots &
	\Gamma_n\end{pmatrix}\\
	& = \begin{pmatrix} \Gamma_1^\top\tilde{A}^{*}\Gamma_1 & \cdots & \Gamma_1^\top\tilde{A}^{*}\Gamma_n \\ \vdots & \ddots & \vdots \\
		\Gamma_n^\top\tilde{A}^{*}\Gamma_1 & \cdots & \Gamma_n^\top\tilde{A}^{*}\Gamma_n
	\end{pmatrix}.
\end{align*}
Hence, 
\[
\tr(\Lambda^2\Gamma\tilde{A}^{*}\Gamma^\top) = \sum_i\lambda_i^2\Gamma_i^\top\tilde{A}^{*}\Gamma_i.
\]
For each $i=1,\ldots,n$, we have $\Gamma_i^\top\tilde{A}^{*}\Gamma_i = \sum_{k}\Gamma_{ik}^2\tilde{A}_{kk}^{*} + \sum_{k\neq l}\Gamma_{ik}\Gamma_{il}\tilde{A}_{kl}^{*}$, and 
\begin{align*}
	\var\left\{\sum_{k\neq l}\Gamma_{ik}\Gamma_{il}\tilde{A}_{kl}^{*}\right\} & = 4\sum_{k<l}\Gamma_{ik}^2\Gamma_{il}^2\var\{\tilde{A}_{kl}^{*}\} \\
	& = 4\sum_{k<l}\Gamma_{ik}^2\Gamma_{il}^2\bE\{(\tilde{A}_{kl}^{*})^2\} \\
	& \leq \dfrac{4}{n}.
\end{align*}
Thus, we have $\Gamma_i^\top\tilde{A}^{*}\Gamma_i = O_p(n^{-1/2})$.

Then, 
\begin{align*}
	\tr(\tilde{A}^{*}(\Delta^{\prime})^2) & = O_p(n^{-1/2})\sum_i\lambda_i^2 \\
	& = O_p(n^{-1/2})\sum_{j,k}\left(\dfrac{\hat{p}_{jk}-p_{jk}}{\sqrt{np_{jk}(1-p_{jk})}}\right)^2 \\
	& = O_p(n^{-1/2})o_p(n^{1/2}) \\
	& = o_p(1).
\end{align*}

Summarizing the above results, we have
\begin{equation}\label{eq:app2}
\tr((\tilde{A}^{\prime})^3-(\tilde{A}^{*})^3)=o_p(1).
\end{equation}

In addition, note that
\[
\tilde{A}_{ij}=\sqrt{\dfrac{np_{ij}(1-p_{ij})}{n\hat{p}_{ij}(1-\hat{p}_{ij})}}\tilde{A}^{\prime}_{ij}\qquad \text{for}\ i\neq j.
\]

Let $\Upsilon\in\bR^{n\times n}$ such that $\Upsilon_{ij}=\sqrt{\dfrac{np_{ij}(1-p_{ij})}{n\hat{p}_{ij}(1-\hat{p}_{ij})}}$. Then, we have $\tilde{A}=\Upsilon\circ\tilde{A}^{\prime}$, where ``$\circ$" denote the Hadmard product of two matrice.

Using  Chernoff bound, we have 
\[
\sqrt{p_{ij}(1-p_{ij})}=\sqrt{\hat{p}_{ij}(1-\hat{p}_{ij})}(1+o_p(n^{-1/4})),
\]
and
\[
\sqrt{\dfrac{np_{ij}(1-p_{ij})}{n\hat{p}_{ij}(1-\hat{p}_{ij})}}=1+o_p(n^{-1/4}).
\]
Then, it is not difficult to obtain that
\[
\tilde{A}=(1+o_p(n^{-1/4}))\tilde{A}^{\prime}
\]

It is simple to verify that
\[
	\tr(\tilde{A}^3)=(1+o_p(n^{-1/4}))^3\tr((\tilde{A}^{\prime})^3).
\]

Then, from \eqref{eq:app2} , we get
\[
\tr(\tilde{A}^3)=(1+o_p(n^{-1/4}))^3(\tr((\tilde{A}^*)^3)+o_p(1)).
\]
This completes the proof of Theorem \ref{thm:null}.

\section*{Funding}
Jianwei Hu is partially supported by the National Natural Science Foundation of China (nos. 12171187, 12371261).

\bibliographystyle{agu04}
\bibliography{refer}

\begin{thebibliography}{53}
\providecommand{\natexlab}[1]{#1}
\expandafter\ifx\csname urlstyle\endcsname\relax
  \providecommand{\doi}[1]{doi:\discretionary{}{}{}#1}\else
  \providecommand{\doi}{doi:\discretionary{}{}{}\begingroup
  \urlstyle{rm}\Url}\fi

\bibitem[{\textit{Airoldi et~al.}(2008)\textit{Airoldi, Blei, Fienberg, and
  Xing}}]{Airoldi:2008}
Airoldi, E.~M., D.~M. Blei, S.~E. Fienberg, and E.~P. Xing (2008), Mixed
  membership stochastic blockmodels, \textit{Journal of machine learning
  research}, \textit{9}(65), 1981--2014, \doi{10.5555/2981780.2981785}.

\bibitem[{\textit{Amini et~al.}(2013)\textit{Amini, Chen, Bickel, and
  Levina}}]{Amini:2013}
Amini, A.~A., A.~Chen, P.~J. Bickel, and E.~Levina (2013), Pseudo-likelihood
  methods for community detection in large sparse networks,, \textit{The Annals
  of Statistics}, \textit{41}(4), 2097--2122, \doi{10.1214/13-AOS1138}.

\bibitem[{\textit{Bai and Silverstein}(2016)}]{Bai:2016}
Bai, Z., and J.~W. Silverstein (2016), \textit{Spectral analysis of large
  dimensional random matrices}, 2nd ed., Springer, New York.

\bibitem[{\textit{Baird and Ulanowicz}(1989)}]{Baird:1989}
Baird, D., and R.~E. Ulanowicz (1989), The seasonal dynamics of the
  {C}hesapeake bay ecosystem, \textit{Ecological Monographs}, \textit{59}(4),
  329–364, \doi{10.2307/1943071}.

\bibitem[{\textit{Bickel and Sarkar}(2016)}]{Bickel:2016}
Bickel, P., and P.~Sarkar (2016), Hypothesis testing for automated community
  detection in networks, \textit{Journal of the Royal Statistical Society:
  Series B (Statistical Methodology)}, \textit{78}(1), 253–273,
  \doi{10.1111/rssb.12117}.

\bibitem[{\textit{Blitzstein and Diaconis}(2011)}]{Blitzstein:2011}
Blitzstein, J., and P.~Diaconis (2011), A sequential importance sampling
  algorithm for generating random graphs with prescribed degrees,
  \textit{Internet Mathematics}, \textit{6}(4), 489–522,
  \doi{10.1080/15427951.2010.557277}.

\bibitem[{\textit{Cammarata and Ke}(2023)}]{Cammarata:2023}
Cammarata, L.~V., and Z.~T. Ke (2023), Power enhancement and phase transitions
  for global testing of the mixed membership stochastic block model,
  \textit{Bernoulli}, \textit{29}(3), 1741--1763, \doi{10.3150/22-BEJ1519}.

\bibitem[{\textit{Chatterjee et~al.}(2011)\textit{Chatterjee, Diaconis, and
  Sly}}]{Chatterjee:2011}
Chatterjee, S., P.~Diaconis, and A.~Sly (2011), Random graphs with a given
  degree sequence, \textit{The Annals of Applied Probability}, \textit{21}(4),
  1400--1435, \doi{10.1214/10-AAP728}.

\bibitem[{\textit{Chen et~al.}(2021)\textit{Chen, Josephs, Lin, Zhou, and
  Kolaczyk}}]{Chen:2021}
Chen, L., N.~Josephs, L.~Lin, J.~Zhou, and E.~D. Kolaczyk (2021), A
  spectral-based framework for hypothesis testing in populations of networks,
  \textit{Statistica Sinica (online)}, \doi{10.5705/ss.202021.0306}.

\bibitem[{\textit{Dong et~al.}(2020)\textit{Dong, Wang, and Liu}}]{Dong:2020}
Dong, Z., S.~Wang, and Q.~Liu (2020), Spectral based hypothesis testing for
  community detection in complex networks, \textit{Information Sciences},
  \textit{512}, 1360–1371, \doi{10.1016/j.ins.2019.10.056}.

\bibitem[{\textit{Du and Tang}(2023)}]{Du:2023}
Du, X., and M.~Tang (2023), Hypothesis testing for equality of latent positions
  in random graphs, \textit{Bernoulli}, \textit{29}(4), 3221--3254,
  \doi{10.3150/22-BEJ1581}.

\bibitem[{\textit{Erd\H{o}s and R\'enyi}(1957)}]{Erdos:1957}
Erd\H{o}s, P., and A.~R\'enyi (1957), On random graphs {I},
  \textit{Publicationes Mathematicae Debrecen}, \textit{6}(290), 290–297.

\bibitem[{\textit{Erd{\H{o}}s et~al.}(2012)\textit{Erd{\H{o}}s, Knowles, Yau,
  and Yin}}]{Erdos:2012}
Erd{\H{o}}s, L., A.~Knowles, H.-T. Yau, and J.~Yin (2012), Spectral statistics
  of {E}rd{\H{o}}s-{R}\'enyi graphs {II}: Eigenvalue spacing and the extreme
  eigenvalues, \textit{Communications in Mathematical Physics volume},
  \textit{314}(3), 587--640, \doi{10.1007/s00220-012-1527-7}.

\bibitem[{\textit{Erd{\H{o}}s et~al.}(2013)\textit{Erd{\H{o}}s, Knowles, Yau,
  and Yin}}]{Erdos:2013}
Erd{\H{o}}s, L., A.~Knowles, H.-T. Yau, and J.~Yin (2013), Spectral statistics
  of {E}rd{\H{o}}s-{R}\'enyi graphs {I}: Local semicircle law, \textit{The
  Annals of Probability}, \textit{41}(3B), 2279--2375, \doi{10.1214/11-AOP734}.

\bibitem[{\textit{Fan et~al.}(2022)\textit{Fan, Fan, Han, and Lv}}]{Fan:2022}
Fan, J., Y.~Fan, X.~Han, and J.~Lv (2022), {SIMPLE}: Statistical inference on
  membership profiles in large networks, \textit{Journal of the Royal
  Statistical Society Series B: Statistical Methodology}, \textit{84}(2),
  630--653, \doi{10.1111/rssb.12505}.

\bibitem[{\textit{Fu and Hu}(2023)}]{Fu:2023-2}
Fu, K., and J.~Hu (2023), Profile-pseudo likelihood methods for community
  detection of multilayer stochastic block models, \textit{Stat},
  \textit{12}(1), e594, \doi{10.1002/sta4.594}.

\bibitem[{\textit{Fu et~al.}(2023)\textit{Fu, Hu, Keita, and Liu}}]{Fu:2023}
Fu, K., J.~Hu, S.~Keita, and H.~Liu (2023), Two-sample test for stochastic
  block models via maximum entry-wise deviation, \textit{Statistics and Its
  Interface (Accepted)}.

\bibitem[{\textit{Fu et~al.}(2024)\textit{Fu, Hu, Keita, and Liu}}]{Fu:2024}
Fu, K., J.~Hu, S.~Keita, and H.~Liu (2024), Two-sample test for stochastic
  block models via the largest singular value, \textit{Communications in
  Statistics - Theory and Methods (online)},
  \doi{10.1080/03610926.2024.2330669}.

\bibitem[{\textit{Ghoshdastidar et~al.}(2020)\textit{Ghoshdastidar, Gutzeit,
  Carpentier, and Luxburg}}]{Ghoshdastidar:2020}
Ghoshdastidar, D., M.~Gutzeit, A.~Carpentier, and U.~V. Luxburg (2020),
  Two-sample hypothesis testing for inhomogeneous random graphs, \textit{The
  Annals of Statistics}, \textit{48}(4), 2208--2229, \doi{1214/19-AOS1884}.

\bibitem[{\textit{Girvan and Newman}(2002)}]{Girvan:2002}
Girvan, M., and M.~E.~J. Newman (2002), Community structure in social and
  biological networks, \textit{Proceedings of the National Academy of
  Sciences}, \textit{99}(12), 7821--7826, \doi{10.1073/pnas.122653799}.

\bibitem[{\textit{Guimer\`{a} and Amaral}(2005)}]{Guimera:2005}
Guimer\`{a}, R., and L.~A.~N. Amaral (2005), Functional cartography of complex
  metabolic networks, \textit{Nature}, \textit{433}(7028), 895–900,
  \doi{1038/nature03288}.

\bibitem[{\textit{Hoff et~al.}(2002)\textit{Hoff, Raftery, and
  Handcock}}]{Hoff:2002}
Hoff, P.~D., A.~E. Raftery, and M.~S. Handcock (2002), Latent space approaches
  to social network analysis, \textit{Journal of the American Statistical
  Association}, \textit{97}(460), 1090--1098,
  \doi{10.1198/0162145023886189061090}.

\bibitem[{\textit{Holland et~al.}(1983)\textit{Holland, Laskey, and
  Leinhardt}}]{Holland:1983}
Holland, P.~W., K.~B. Laskey, and S.~Leinhardt (1983), Stochastic blockmodels:
  First steps, \textit{Social Networks}, \textit{5}(2), 109--137,
  \doi{10.1016/0378-8733(83)90021-7}.

\bibitem[{\textit{Hu et~al.}(2020)\textit{Hu, Qin, Yan, , and Zhao}}]{Hu:2020}
Hu, J., H.~Qin, T.~Yan, , and Y.~Zhao (2020), Corrected bayesian information
  criterion for stochastic block models, \textit{Journal of the American
  Statistical Association}, \textit{115}(532), 1771--1783,
  \doi{10.1080/01621459.2019.1637744}.

\bibitem[{\textit{Hu et~al.}(2021)\textit{Hu, Zhang, Qin, Yan, and
  Zhu}}]{Hu:2021}
Hu, J., J.~Zhang, H.~Qin, T.~Yan, and J.~Zhu (2021), Using maximum entry-wise
  deviation to test the goodness of fit for stochastic block models,
  \textit{Journal of the American Statistical Association}, \textit{116}(535),
  1373--1382, \doi{10.1080/01621459.2020.1722676}.

\bibitem[{\textit{Jin}(2015)}]{Jin:2015}
Jin, J. (2015), Fast community detection by score, \textit{The Annals of
  Statistics}, \textit{43}(1), 57--89, \doi{10.1214/14-AOS1265}.

\bibitem[{\textit{Jin et~al.}(2021)\textit{Jin, Ke, and Luo}}]{Jin:2021}
Jin, J., Z.~T. Ke, and S.~Luo (2021), Optimal adaptivity of signed-polygon
  statistics for network testing, \textit{The Annals of Statistics},
  \textit{49}(6), 3408--3433, \doi{10.1214/21-AOS2089}.

\bibitem[{\textit{Jin et~al.}(2023)\textit{Jin, Ke, and Luo}}]{Jin:2023}
Jin, J., Z.~T. Ke, and S.~Luo (2023), Mixed membership estimation for social
  networks, \textit{Journal of Econometrics (online)},
  \doi{10.1016/j.jeconom.2022.12.003}.

\bibitem[{\textit{Karrer and Newman}(2011)}]{Karrer:2011}
Karrer, B., and M.~E.~J. Newman (2011), Stochastic blockmodels and community
  structure in networks, \textit{Physical Review. E}, \textit{83}(1), 016,107,
  \doi{10.1103/PhysRevE.83.016107}.

\bibitem[{\textit{Ke and Wang}(2022)}]{Ke:2022}
Ke, Z.~T., and J.~Wang (2022), Optimal network membership estimation under
  severe degree heterogeneity, \textit{arXiv:2204.12087}.

\bibitem[{\textit{Lei}(2016)}]{Lei:2016}
Lei, J. (2016), A goodness-of-fit test for stochastic block models, \textit{The
  Annals of Statistics}, \textit{44}(1), 401--424, \doi{10.1214/15-AOS1370}.

\bibitem[{\textit{Lusseau et~al.}(2003)\textit{Lusseau, Schneider, Boisseau,
  Haase, Slooten, and Dawson}}]{Lusseau:2003}
Lusseau, D., K.~Schneider, O.~Boisseau, P.~Haase, E.~Slooten, and S.~Dawson
  (2003), The bottlenose dolphin community of doubtful sound features a large
  proportion of long-lasting associations, \textit{Behavioral Ecology and
  Sociobiology}, \textit{54}(4), 396--405, \doi{10.1007/s00265-003-0651-y}.

\bibitem[{\textit{Mao et~al.}(2021)\textit{Mao, Sarkar, and
  Chakrabarti}}]{Mao:2021}
Mao, X., P.~Sarkar, and D.~Chakrabarti (2021), Estimating mixed memberships
  with sharp eigenvector deviations, \textit{Journal of the American
  Statistical Association}, \textit{116}(536), 1928--1940,
  \doi{10.1080/01621459.2020.1751645}.

\bibitem[{\textit{Mukherjee et~al.}(2018)\textit{Mukherjee, Mukherjee, and
  Sen}}]{Mukherjee:2018}
Mukherjee, R., S.~Mukherjee, and S.~Sen (2018), Detection thresholds for the
  $\beta$-model on sparse graphs, \textit{The Annals of Statistics},
  \textit{46}(3), 1288--1317, \doi{10.1214/17-AOS1585}.

\bibitem[{\textit{Nickel}(2008)}]{Nickel:2008}
Nickel, C. L.~M. (2008), Random dot product graphs a model for social networks,
  Ph.D. thesis, Johns Hopkins University.

\bibitem[{\textit{Rinaldo et~al.}(2013)\textit{Rinaldo, Petrovi\'{c}, and
  Fienberg}}]{Rinaldo:2013}
Rinaldo, A., S.~Petrovi\'{c}, and S.~E. Fienberg (2013), Maximum lilkelihood
  estimation in the $\beta$-model, \textit{The Annals of Statistics},
  \textit{41}(3), 1085--1110, \doi{10.1214/12-AOS1078}.

\bibitem[{\textit{Rohe et~al.}(2011)\textit{Rohe, Chatterjee, and
  Yu}}]{Rohe:2011}
Rohe, K., S.~Chatterjee, and B.~Yu (2011), Spectral clustering and the
  high-dimensional stochastic blockmodel, \textit{The Annals of Statistics},
  \textit{39}(4), 1878--1915, \doi{10.1214/11-AOS887}.

\bibitem[{\textit{Rubin-Delanchy et~al.}(2022)\textit{Rubin-Delanchy, Cape,
  Tang, and Priebe}}]{Rubin:2022}
Rubin-Delanchy, P., J.~Cape, M.~Tang, and C.~E. Priebe (2022), A statistical
  interpretation of spectral embedding: The generalised random dot product
  graph, \textit{Journal of the Royal Statistical Society Series B: Statistical
  Methodology}, \textit{84}(4), 1446--1473, \doi{10.1111/rssb.12509}.

\bibitem[{\textit{Sald{\~n}a et~al.}(2017)\textit{Sald{\~n}a, Yu, and
  Feng}}]{Saldna:2017}
Sald{\~n}a, D.~F., Y.~Yu, and Y.~Feng (2017), How many communities are there?,
  \textit{Journal of Computational and Graphical Statistics}, \textit{26}(1),
  171--181, \doi{10.1080/10618600.2015.1096790}.

\bibitem[{\textit{Scott}(2000)}]{Scott:2000}
Scott, J. (2000), \textit{Social network analysis: {A} handbook}, 2nd ed.,
  SAGE, London.

\bibitem[{\textit{Sussman et~al.}(2014)\textit{Sussman, Tang, and
  Priebe}}]{Sussman:2014}
Sussman, D.~L., M.~Tang, and C.~E. Priebe (2014), Consistent latent position
  estimation and vertex classification for random dot product graphs,
  \textit{IEEE Transactions on Pattern Analysis and Machine Intelligence},
  \textit{36}(1), 48--57, \doi{10.1109/TPAMI.2013.135}.

\bibitem[{\textit{Tang et~al.}(2017)\textit{Tang, Athreya, D.~L.~Sussman, Park,
  and Priebe}}]{Tang:2017}
Tang, M., A.~Athreya, V.~L. D.~L.~Sussman, Y.~Park, and C.~E. Priebe (2017), A
  semiparametric two-sample hypothesis testing problem for random graphs,
  \textit{Journal of Computational and Graphical Statistics}, \textit{6}(2),
  344--354, \doi{10.1080/10618600.2016.1193505}.

\bibitem[{\textit{Wang et~al.}(2023)\textit{Wang, Zhang, Liu, Zhu, and
  Guo}}]{Wang:2023}
Wang, J., J.~Zhang, B.~Liu, J.~Zhu, and J.~Guo (2023), Fast network community
  detection with profile-pseudo likelihood methods, \textit{Journal of the
  American Statistical Association}, \textit{118}(542), 1359--1372,
  \doi{10.1080/01621459.2021.1996378}.

\bibitem[{\textit{Wang and Bickel}(2017)}]{Wang:2017}
Wang, Y.~R., and P.~J. Bickel (2017), Likelihood-based model selection for
  stochastic block models, \textit{The Annals of Statistics}, \textit{45}(2),
  500--528, \doi{10.1214/16-AOS1457}.

\bibitem[{\textit{Wang and Yao}(2021)}]{Wang:2021}
Wang, Z., and J.~Yao (2021), On a generalization of the clt for linear
  eigenvalue statistics of wigner matrices with inhomogeneous fourth moments,
  \textit{Random Matrices: Theory and Applications}, \textit{10}(4), 2150,041,
  \doi{10.1142/S2010326321500416}.

\bibitem[{\textit{Westveld and Hoff}(2011)}]{Westveld:2011}
Westveld, A.~H., and P.~D. Hoff (2011), A mixed effects model for longitudinal
  relational and network data with applications to international trade and
  conflict, \textit{The Annals of Applied Statistics}, \textit{5}(2), 843--872,
  \doi{10.1214/10-AOAS403}.

\bibitem[{\textit{Wu et~al.}(2022)\textit{Wu, Kong, and Xu}}]{Wu:2022}
Wu, F., X.~Kong, and C.~Xu (2022), Test on stochastic block model: local
  smoothing and extreme value theory, \textit{Journal of Systems Science and
  Complexity}, \textit{35}(4), 1535--1556, \doi{10.1007/s11424-021-0154-9}.

\bibitem[{\textit{Wu and Hu}(2024)}]{Wu:2024}
Wu, Q., and J.~Hu (2024), A spectral based goodness-of-fit test for stochastic
  block models, \textit{Statistics and Probability Letters}, \textit{209},
  110,104, \doi{10.1016/j.spl.2024.110104}.

\bibitem[{\textit{Xie and Xu}(2020)}]{Xie:2020}
Xie, F., and Y.~Xu (2020), Optimal bayesian estimation for random dot product
  graphs, \textit{Biometrika}, \textit{107}(4), 875–889,
  \doi{10.1093/biomet/asaa031}.

\bibitem[{\textit{Xie and Xu}(2023)}]{Xie:2023}
Xie, F., and Y.~Xu (2023), Efficient estimation for random dot product graphs
  via a one-step procedure, \textit{Journal of the American Statistical
  Association}, \textit{118}(541), 651–664,
  \doi{10.1080/01621459.2021.1948419}.

\bibitem[{\textit{Yan and Xu}(2013)}]{Yan:2013}
Yan, T., and J.~Xu (2013), A central limit theorem in the $\beta$-model for
  undirected random graphs with a diverging number of vertices,
  \textit{Biometrika}, \textit{100}(2), 519--524, \doi{10.1093/biomet/ass084}.

\bibitem[{\textit{Zachary}(1977)}]{Zachary:1977}
Zachary, W.~W. (1977), An information flow model for conflict and fission in
  small groups, \textit{Journal of Anthropological Research}, \textit{33}(4),
  452--473, \doi{10.1086/jar.33.4.3629752}.

\bibitem[{\textit{Zhao et~al.}(2012)\textit{Zhao, Levina, and Zhu}}]{Zhao:2012}
Zhao, Y., E.~Levina, and J.~Zhu (2012), Consistency of community detection in
  networks under degree-corrected stochastic block models, \textit{The Annals
  of Statistics}, \textit{40}(4), 2266--2292, \doi{10.1214/12-AOS1036}.

\end{thebibliography}

\end{document}